\documentclass[conference]{IEEEtran}
\IEEEoverridecommandlockouts
\newcommand{\mi}{\mathrm{i}}

\usepackage{paralist}

\usepackage{amssymb}  % assumes amsmath package installed

\usepackage{ stmaryrd }
\usepackage{amsthm}
\usepackage[export]{adjustbox}
\usepackage[section]{placeins}
\usepackage{graphics,subfigure} % for pdf, bitmapped graphics files
\usepackage{epsfig} % for postscript graphics files
\usepackage{epstopdf}
\usepackage{mathptmx} % assumes new font selection scheme installed
\usepackage{times} % assumes new font selection scheme installed
\usepackage{amsmath} % assumes amsmath package installed
\usepackage{floatrow}
\usepackage{multirow}
\usepackage{multicol}
\usepackage{mathtools}
\usepackage{float}
\usepackage{xcolor}
\usepackage{upgreek}
\usepackage{ mathrsfs }
\usepackage{ bbm }
\usepackage{bm}
\usepackage{cleveref}
\DeclareMathAlphabet{\mathcal}{OMS}{cmsy}{b}{n}
\DeclareMathAlphabet{\mathcal}{OMS}{cmsy}{m}{n}
\usepackage{lipsum}
\usepackage{tabularx, calc}
\newtheorem{lemma}{\indent Lemma}

\newtheorem{proposition}{\indent Proposition}
\newtheorem{definition}{\indent Definition}

\newtheorem{example}{\indent Example}
\newtheorem{remark}{\indent Remark}
\newtheorem{problem}{\indent Problem}

\usepackage{cite}
\usepackage{amsmath,amssymb,amsfonts}
\usepackage{algorithmic}
\usepackage{graphicx}
\usepackage{textcomp}
\def\BibTeX{{\rm B\kern-.05em{\sc i\kern-.025em b}\kern-.08em
    T\kern-.1667em\lower.7ex\hbox{E}\kern-.125emX}}
\begin{document}
	\title{Generation of accessible sets in the dynamical modelling of quantum network systems*}
\author{Qi Yu, %~\IEEEmembership{Member,~IEEE,}
	Yuanlong Wang, Daoyi~Dong,
	~Ian~R.~Petersen,
	Guo-Yong Xiang
	%Muhammad Fuady Emzir
	
	\thanks{
		This work was supported by the Australian Research Council's Discovery Projects funding scheme under Projects DP190101566 and DP180101805, the U.S. Office of Naval Research Global under Grant N62909-19-1-2129 and the Air Force Office of Scientific Research and the Office of Naval Research Grants under agreement number FA2386-16-1-4065.}
	%	\thanks{Qi Yu is with the Research School of Engineering, Australian National University, Canberra, ACT 2601, Australia and the School of Engineering and Information Technology, University of New South Wales, Canberra, ACT 2600, Australia (e-mail: vickivicky.qi.yu@gmail.com).}
		\thanks{Qi Yu and Daoyi Dong are with the School of Engineering and Information Technology, University of New South Wales, Canberra, ACT 2600, Australia (e-mail:vickivicky.qi.yu@gmail.com; daoyidong@gmail.com).}% <-this % stops a space
	\thanks{Yuanlong Wang is with the Centre for Quantum Dynamics, Griffith University, Brisbane, QLD 4111, Australia and the School of Engineering and Information Technology, University of New South Wales, Canberra, ACT 2600, Australia (e-mail: yuanlong.wang.qc@gmail.com). }
	
	\thanks{Ian R. Petersen is with the Research School of Electrical, Energy and Materials Engineering, Australian National University, Canberra, ACT 2601, Australia (e-mail: i.r.petersen@gmail.com). }
	\thanks{Guo-Yong Xiang is with the CAS Key Laboratory of Quantum Information and CAS Center For Excellence in Quantum Information and Quantum Physics, University of Science and Technology of China, Hefei, China (e-mail: gyxiang@ustc.edu.cn). }
}

\iffalse	
	\thanks{This work was supported by the Australian Research Council's Discovery Projects funding scheme under Projects DP190101566 and DP180101805, the U.S. office of Naval Research Global under Grant N62909-19-1-2129 and the Air Force Office of Scientific Research and the Office of Naval Research Grants under agreement number FA2386-16-1-4065.}

\author{\IEEEauthorblockN{Qi Yu}
\IEEEauthorblockA{\textit{ School of Engineering and Information Technology} \\
\textit{ University of New South Wales}\\
Canberra, Australia \\
vickivicky.qi.yu@gmail.com}
\and
\IEEEauthorblockN{Yuanlong Wang}
\IEEEauthorblockA{\textit{School of Engineering and Information Technology} \\
	\textit{University of New South Wales}\\
	Canberra, Australia \\
	\textit{Centre for Quantum Dynamics} \\
	\textit{Griffith University}\\
	Brisbane, QLD 4111, Australia \\
	yuanlong.wang.qc@gmail.com}
\and
\IEEEauthorblockN{Daoyi Dong}
\IEEEauthorblockA{\textit{ School of Engineering and Information Technology} \\
\textit{University of New South Wales}\\
Canberra, Australia \\
daoyidong@gmail.com}
\and
\IEEEauthorblockN{ Ian R. Petersen}
\IEEEauthorblockA{\textit{ Research School of Engineering} \\
\textit{Australian National University}\\
Canberra, Australia \\
i.r.petersen@gmail.com}
}
\fi
\maketitle

\begin{abstract}
In this paper, we consider the dynamical modeling of a class of quantum network systems consisting of qubits. Qubit probes are employed to measure a set of selected nodes of the quantum network systems. For a variety of applications, a state space model is a useful way to model the system dynamics. To construct a state space model for a quantum network system, the major task is to find an accessible set containing all of the operators coupled to the measurement operators. This paper focuses on the generation of a proper accessible set for a given system and measurement scheme. We provide analytic results on simplifying the process of generating accessible sets for systems with a time-independent Hamiltonian. Since the order of elements in the accessible set determines the form of state space matrices, guidance is provided to effectively arrange the ordering of elements in the state vector. Defining a system state according to the accessible set, one can develop a state space model with a special pattern inherited from the system structure. As a demonstration, we specifically consider a typical 1D-chain system with several common measurements, and employ the proposed method to determine its accessible set.

\end{abstract}


\section{INTRODUCTION}
The dynamical modeling of quantum systems is a basic task for a variety of quantum engineering problems such as quantum identification \cite{yuanlong2018Algorithm,Akira2017Hamiltonian,Wang2019Gate,Jun2015identifi,Pan2017Dark,Sone2018Quantify,Sone2017Exact,yuanlong2020,Guofeng2015realization,Burgarth2009indirect,Bonnabel2009Observer,Degen2017,qibo2017adaptive,Levitt2017identification}, quantum filtering \cite{qiyu2019TCST,Qing2016Coherent,qiyu2018SMC,qing2019design}, quantum control \cite{Cui2019Modeling,Daoyi2019Cybernetics,WS2016free,JS2011Ensemble,Xiang2017Performance,Kuang2017Rapid,Rebin2016Spatial,Shi2017Fault,Yanan2017Lyapunov,yuguo2019vanishing}. A good dynamical model can benefit the analysis of these problems. This paper studies the modeling of a class of quantum network systems whose element systems are qubits and the structure of the system Hamiltonians is given \cite{Yuzuru2014structure,Shu2017identify,	Ticozzi2015consensus,Ticozzi2014network,Guodong2015TAC}. The ultimate objective is to generate a state space model for a quantum network system subject to a measurement scheme. To find the state space equations for the system, a key task is to generate an accessible set of operators that are coupled with the measurement operators \cite{Yuan2006Reachable}. Once an accessible set is obtained, the system state vector consists of the expectation values of all operators in the accessible set. The state space equations can then be deduced given the state vector and the system Hamiltonian.
%The network system is chosen to be a chain system consisting an arbitrary number of qubit systems. 

The generation of accessible sets is usually complicated. For most cases, the number of elements in an accessible set increases rapidly with the number of subsystems in a network system (See Fig. \ref{quantum Network}) \cite{Yuzuru2014structure}, and thus it may be difficult to search for numerical solutions in high-dimensional systems. Although one can always turn to a computer for solutions, the computational complexity can be high. Moreover, the ordering of the elements in the state vector is also nontrivial. Arranging a good ordering of elements in the system state variable may lead to state space matrices with a good structure. In the conference paper \cite{Qi2019SMC}, preliminary results have been presented in searching for a rapid method for the generation of accessible sets. This paper aims at presenting a comprehensive investigation on obtaining good accessible sets while simplifying the generation process. The specific definition of ``good'' is to give to a state space matrix that is easy to analyze and has a repetition pattern as the qubit number increases. 

 We first generalize the generation rules to achieve a lower computational complexity. Then we provide several lemmas and propositions to further reduce the computational complexity for a class of spin chain systems. We employ graphs to describe the generation of accessible sets. The graph method is a powerful tool for the demonstration of generation processes. We prove that the generation of accessible sets can be decomposed as the generation of a series of subsets for a class of quantum chain systems. The division of graphs can help in revealing the repetition pattern of the state space matrices. Graphs can also provide a guidance for the ordering of elements in the state vector. A state space model for the quantum network system can be immediately obtained given the corresponding accessible set. 

The structure of this paper is as follows: Section \ref{Sec2} formulates the problem. Section \ref{Sec3} presents our main results. A series of illustrative examples are given in Section \ref{Sec4}. Section \ref{Conclusion} concludes this paper.

%==================  Problem Formulation    ==================================

\section{Problem Formulation}\label{Sec2}

\subsection{State space equations and accessible sets}

Measurement is often needed to extract information about a quantum network system. However, limited by experimental devices, it is common that only part of the network system can be measured in many practical applications (See Fig. \ref{quantum Network}). For example, one can measure one or two nodes at one edge of the network system to infer information about the whole system. 

\begin{figure}	
	\centering		
	\includegraphics[width=5.5cm]{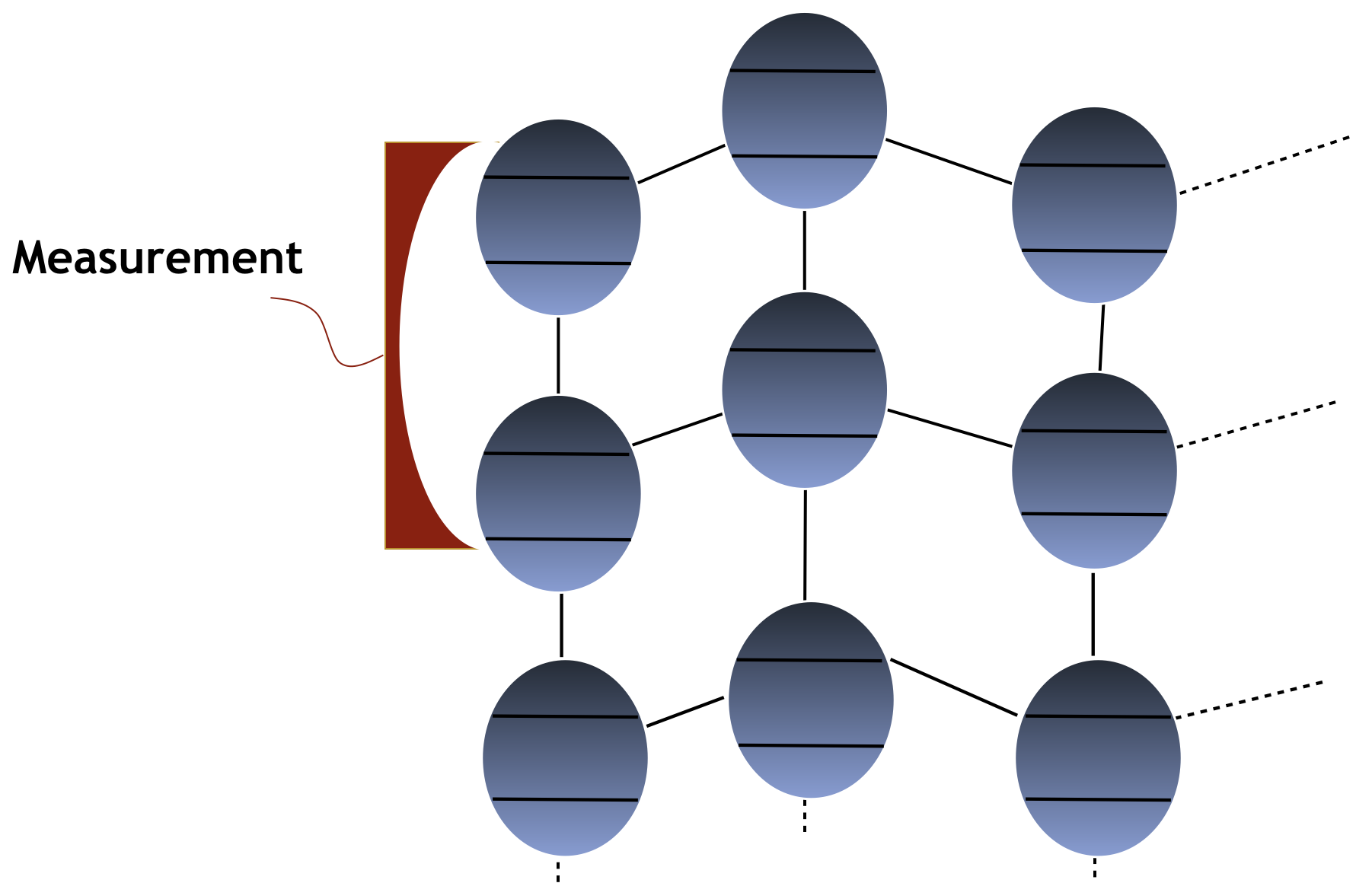}		
	\caption{An example of a quantum network system. The nodes are qubits and a connecting line indicates coupling between two qubits. A measurement device is employed to measure two nodes at an edge of the system.}
	\label{quantum Network}		
\end{figure}

Given $H$ as the time-independent system Hamiltonian, the time evolution of an arbitrary system observable $O(t)$ in the Heisenberg picture is 
\begin{equation}\label{eq01timevolution}
O(t)=U^\dagger(t)OU (t)
\end{equation}
where $U(t)$ is the unitary operator with time evolution 
\begin{equation}
U(t)=e^{-\mi Ht}.
\end{equation}
Here, $\mi$ is the imaginary unit (we have set $\hbar=1$). 
Taking derivative of both sides of \eqref{eq01timevolution}, we have 
\begin{equation}\label{Sfunction}
\frac{dO(t)}{dt}=\mi[H,O(t)].
\end{equation}
Given a measurement operator $M(0)=M$, its time evolution is 
\begin{equation}
M(t)=e^{\mi Ht}Me^{-\mi Ht}.
\end{equation}
According to the Baker-Hausdorff Lemma 
\cite{Sakurai2005modern}, the Taylor series of $M(t)$ is
\begin{equation}\label{MeasureEvolution}
\begin{split}
M(t)=&M+[H,M]\mi t+[H,[H,M]]\frac{(\mi t)^2}{2!} \\
&+[H,[H,[H,M]]]\frac{(\mi t)^3}{3!}+\ \cdots
\end{split}
\end{equation}
According to \eqref{MeasureEvolution}, the time derivatives of the measurement operator $M(t)$ are given
\begin{equation}\label{derivitives}
\mi[H,M],\  -[H,[H,M]],\  -\mi[H,[H,[H,M]]],\  \cdots .
\end{equation}
Since a state space model only contains first order derivatives of the elements of the state vector, we need to find a set $G$ of basis operators for the derivatives given in \eqref{derivitives}. We refer to the set $G$ as the \textit{accessible set} corresponding to the measurement $M$ since all of the element operators in $G$ are accessible by the measurement $M$. In other words, $G$ is a set of operators whose dynamics are coupled with $M$. A set of rules to generate the accessible sets is given in \cite{junzhang2014HIdentifi}.

Suppose the set $G$ has already been obtained and is given as follows:
\begin{equation}
G=\{O_1,\  O_2, \ O_3,\  \cdots, \  O_{N_o}\}
\end{equation}
where $N_o$ is the number of operators in $G$. We then summarize the process of generating the state space model when given an accessible set. We define the system state vector $\mathbf{x}$ as
\begin{equation}\label{stateformulation}
\mathbf{x}=(\hat{O}_1, \   \hat{O}_2, \  \hat{O}_3, \  \cdots, \  \hat{O}_{N_o})^T,
\end{equation}
where $O_k$ is the $k$-th operator in $G$ and $\hat{O}_k=\text{Tr}(O_k\rho)$ is the expectation of observable $O_k$.  As shown in \cite{junzhang2014HIdentifi}, the following state space equations can be employed to describe the dynamics of the system state $\mathbf{x}$ and the measurement $\mathbf{y}=\hat{M}$
\begin{equation}\label{statespaceeq}
\left\{
\begin{array}{ll}
\dot{\mathbf{x}} &=A\mathbf{x}+B\mathbf{x_0}, \\
\mathbf{y} &=C\mathbf{x},
\end{array}
\right.
\end{equation}
where $A$, $B$ and $C$ are coefficient matrices that can be obtained using \eqref{Sfunction}.

We divide the task of deriving a state space model in \eqref{statespaceeq} for a quantum network system into two parts: The first is to find an accessible set so as to define a state vector $\mathbf{x}$; The second is to find the coefficient matrices $A$, $B$ and $C$ once the state vector $\mathbf{x}$ is determined. The matrix $A$ can be calculated using \eqref{Sfunction}, $B$ depends on the initial state, $C$ depends on the measurement operators and $\mathbf{x}_0$ is the initial state. In this paper, we mainly focus on the first step since the second step is straightforward after obtaining a proper accessible set.

 Although accessible set is important for constructing state space model, the generation of the accessible set is not easy except for systems with simple coupling structures and special measurement schemes. For general cases, the difficulty of generating accessible sets increases rapidly with the number of qubits in the network system. Moreover, note that the ordering of the elements forming the state $\mathbf{x}$ in \eqref{stateformulation} determines the structure of the matrices $A$, $B$ and $C$. A good ordering should have the following properties:
\begin{enumerate}
	\item The matrix $A$ has a structure that can simplify further analysis.
	\item The matrix $A$ possesses a repetition pattern which is straightforward to extend when the number of qubits in the quantum network system increases.
\end{enumerate}
In this paper, we mainly study the generation of accessible sets. Our goal is to simplify the generation processes given in \cite{junzhang2014HIdentifi} while obtaining a good ordering for elements in the state vector $\mathbf{x}$.

\subsection{Problem formulation}
We assume that the Hamiltonian of a quantum network system consisting of $N$ qubits takes the following form:
\begin{equation}\label{hamiltonian}
H=\sum_{k=1}^{N_f} h_k H_k,
\end{equation}
where $\{H_k\}$ are Hermitian operators depending on the way the qubits coupled with each other and $\{h_k\}$ are coupling strengths. $N_f$ is the number of unknown parameters. Let $\digamma$ denote the set of operators constructing the Hamiltonian, which takes the following form
\begin{equation}\label{Hset}
\digamma = \{ H_1, \  H_2,\  H_3,\   \cdots,\  H_k, \  \cdots, \  H_{N_f} \}.
\end{equation}
We call $\digamma$ the \textit{Hamiltonian set} of $H$.

%\begin{definition}\label{requirementsonLambda}
%	$\Lambda$ is the complete basis set of all operators $\{O_i\}$ of the network system such that
%	\begin{enumerate}
%		\item $O_j=\otimes_k O^q_k$ where $O^q_k$ is an operator on the $k$-th qubit of the network system.
%		\item $\sum f_j O_j \in \Lambda \quad \forall O_j \in \Lambda,\  f_j\in \mathbb{R}$;
%		\item  There exists a set $\{O_j\}$ which forms a complete set of basis for $\Lambda$.
%	\end{enumerate}
%\end{definition}

We generalize $M$ to be a set for all applicable measurement operators as
\begin{equation}\label{MeasurementSet}
M=\{O_{1}, \ O_{2},\  \cdots\}.
\end{equation} 
Given $M$ and the Hamiltonian $H$, to find the accessible set $G$, we set the initial accessible set as $G_0=M$. Then, we iteratively update the accessible set using the following rule until saturated \cite{junzhang2014HIdentifi}:
\begin{equation}\label{iterativeAS}
	G_m=\llbracket G_{m-1},\digamma \rrbracket \cup G_{m-1},
\end{equation}
where 
\begin{equation}\label{RuleForAS}
	\llbracket G_{m-1},\digamma \rrbracket=\{O_j|\text{Tr}(O_j^\dagger [\tau,\nu])\neq 0, \exists \tau \in G_{m-1}, \nu \in \digamma, O_j \in \Lambda \}.
\end{equation}
We use $\Lambda$ to represent a complete basis set for Hermitian operators of a qubit network system. For an $N$-qubit system, the number of operators in $\Lambda$ is $3^N$. The generation rule \eqref{RuleForAS} indicates that, finding an accessible set involves finding all of the operators coupled with the measurement operators in \eqref{MeasurementSet}. 

The following definition is used for a concise presentation.
\begin{definition}
	Given a triplet $\{\Lambda,\digamma, M\}$, where $\digamma$ is a Hamiltonian set and $M$ is a measurement set, the function $f$ is defined as
	$f:\{\Lambda,\digamma, M\} \rightarrow G$ where $G$ is the accessible set generated by the triplet $\{\Lambda,\digamma, M\}$.
\end{definition}

We formulate our problem as follows:
\begin{problem}\label{initialProblem}
	Let $G=f(\Lambda,\digamma, M)$ where $\digamma$ is the Hamiltonian set given in \eqref{Hset} and $M$ is the measurement set given in \eqref{MeasurementSet}, we aim to develop an economic method to simplify the generation of the accessible set $G$ with a good ordering according to generation rules \eqref{iterativeAS} and \eqref{RuleForAS}. 
\end{problem}
In Problem \ref{initialProblem}, the set $\Lambda$ scales exponentially. Thus, an algorithm can be time-consuming since it may require a full search of $\Lambda$, accompanying a high probability to yield an accessible set with an unsatisfactory ordering. Our study aims to investigate Problem \ref{initialProblem} for generating a good accessible set efficiently.

\section{Main results}\label{Sec3}

%===================      theoretical results    ==============================================
In this section, we first simplify the generation rules \eqref{iterativeAS} and \eqref{RuleForAS} to reduce the computational complexity. We then propose a method to achieve a good ordering for accessible sets. We also provide several lemmas and propositions that can help the calculation.

\subsection{Regarding the computational complexity}\label{SecComplex}
Define $\Omega$ as the set of all operators that are the tensor product of $N$ Pauli matrices and the identity. We have
\begin{equation}\label{Omega}
\Omega=\{O|O=\sigma_{i_1}\otimes\sigma_{i_2}\  \cdots\  \otimes\sigma_{i_k}\otimes\  \cdots\  \sigma_{i_N}\}
\end{equation}
where $\otimes$ denotes tensor product, $ i_k \in \{0,1,2,3\}$ and 
\begin{equation}\label{pauli}
\begin{split}
\sigma_0&\coloneqq  I_{2\times 2}, \qquad\qquad\qquad\ 
\sigma_1\coloneqq\sigma_x=\left(
\begin{array}{cc}
0 & 1\\
1 & 0\\
\end{array}\right),\\
\sigma_2&\coloneqq\sigma_y=\left(
\begin{array}{cc}
0 & -i\\
i & 0\\
\end{array}\right), \quad
\sigma_3\coloneqq\sigma_z=\left(
\begin{array}{cc}
1 & 0\\
0 & -1\\
\end{array}\right).
\end{split}
\end{equation}
We also have the equality $\sigma_{i_k}^2=I$.

The set $\Omega$ is an unnormalized basis set of the operator space for the network system. For the rest of the paper, we work with the set $\Omega$ rather than $\Lambda$ for the generation of accessible sets.

\begin{definition}\label{decomposition}
	An operator set $\bar{S}\subset \Omega$ is defined as the decomposed set of $S$, if $\bar{S}$ is a minimal basis set of $S$.
\end{definition}
\begin{remark} Any measurement set $M$ can be decomposed to a corresponding measurement set $\bar{M}\subset \Omega$. For the set $\digamma$, we can also do the decomposition to make $\bar{\digamma}\subset \Omega$. The set $\bar{G}=f(\Omega,\bar{\digamma},\bar{M})$ is the decomposed form of $G=f(\Omega,\digamma,M)$. Though the state space models based on $G$ and $\bar{G}$ can be different in format, they are equivalent in describing the same system dynamics. %For the rest of the paper, we assume both the set $M$ and $\digamma$ are decomposed.
\end{remark}
\begin{definition}\label{CommuteOmega}
	The operation $\lfloor \cdot,\cdot \rceil$ is defined on any operators $A$ and $B$ such that
    $\lfloor A,B \rceil=O$ where $O\in \Omega$ and $O\propto [A,B]$.
\end{definition}

The following proposition simplifies the generation of accessible sets for qubit network systems.
\begin{proposition}\label{SimplifyPropo}
	For a qubit network system $\{\Omega,\bar{\digamma},\bar{M}\}$, the generation rules \eqref{iterativeAS} and \eqref{RuleForAS} are equivalent to the following rule
%\begin{equation}
%\llbracket     1 \rrbracket       \quad
%\llparenthesis 2 \rrparenthesis   \quad
%\llceil        3 \rrceil          \quad
%\llfloor       4 \rrfloor         \quad
%	[\![ e=mc^2 ]\!]
%\end{equation}
	\begin{equation}\label{NewRuleForAS01}
		G_m=\llparenthesis G_{m-1},\bar{\digamma} \rrparenthesis \cup G_{m-1}, 
	\end{equation}
	where
	\begin{equation}\label{NewRuleForAS}
	\begin{split}
	\llparenthesis  G_{m-1},\bar{\digamma}\rrparenthesis = & \{O_{\tau,\nu}|O_{\tau,\nu}=\lfloor \tau,\nu\rceil,\\
	&O_{\tau,\nu}\neq 0, \tau \in G_{m-1}, \nu \in \bar{\digamma}]\}.
	\end{split}
	\end{equation}.
%	Here, $\mathcal{C}_{\tau,\nu}$ is some proper nonzero coefficient to make $\mathcal{C}_{\tau,\nu}[\tau,\nu]\in \Omega$.

\end{proposition}

\begin{proof}
	The Pauli matrices are orthogonal in the sense
	\begin{equation}\label{orthogonal}
	\text{Tr}(\sigma_a^\dagger \sigma_b)=
	\begin{cases}
	2 & a=b, \\
	0  & a\neq b,
	\end{cases}
	\end{equation}
	where $a,b \in \{0,1,2,3\}$. 
	
	Note that the Pauli matrices obey the following commutation relations 
	\begin{equation}\label{PauliCommutaRelation}
	[\sigma_a,\sigma_b]=
	\begin{cases}
	2\mi\epsilon_{abc}\sigma_c & a\neq b,\\
	0                        & a=b,
	\end{cases}
	\end{equation}
	where $a,b,c \in \{1,2,3\}$ and the constant $\epsilon_{abc}$ is the Levi-Civita symbol. Equation \eqref{PauliCommutaRelation} indicates that the commutator of Pauli matrices yields either a matrix that is proportional to another Pauli matrix or $0$. Based on this fact, we have $$\mathcal{C}_{O_{1,2}}[O_1,O_2]\in \Omega \quad \forall  O_1,O_2\in \Omega,$$ where $\mathcal{C}_{O_{1,2}}$ is a proper nonzero coefficient. 
	We can conclude that there exists a proper nonzero coefficient $\mathcal{C}_{\tau,\nu}$ such that
	\begin{equation}\label{AllInLambda}
	\mathcal{C}_{\tau,\nu}[\tau,\nu]\in \Omega \quad \forall \tau \in G_0, \nu \in \digamma.
	\end{equation}
	Equation \eqref{AllInLambda} indicates that the accessible set $G_m\subset \Omega$ given that $G_{m-1}\subset \Omega$ using generation rule \eqref{NewRuleForAS}. In our case, we have $\bar{\digamma} \subset \Omega$ and $\bar{G}_0 \subset \Omega$, which assures that the generation rule \eqref{NewRuleForAS} can guarantee that $G\in\Omega$.
		Let 
	\begin{equation}\label{Otao}
	O_{\tau,\nu}=\mathcal{C}_{\tau,\nu}[\tau,\nu],
	\end{equation}
	we have $O_{\tau,\nu}\in\Omega$, which confirms that the commutator of two operators in $\Omega$ yields another operator that is proportional to an operator in $\Omega$. According to \eqref{PauliCommutaRelation}, if $O_{\tau,\nu}\neq 0$, then
	\begin{equation}\label{OinApp}
	\text{Tr}(O_{\tau,\nu}^\dagger [\tau,\nu])\neq 0.
	\end{equation}
	According to \eqref{orthogonal}, for any $O\in\Omega$ with $O \neq O_{\tau,\nu}$,
	we have
	\begin{equation}\label{OzeroApp}
	\text{Tr}(O^\dagger [\tau,\nu])= 0.
	\end{equation}
	Equations \eqref{OinApp} and \eqref{OzeroApp} together indicate that $O_{\tau,\nu}$ is the operator that satisfies the requirement in \eqref{RuleForAS} and thus should be added into the accessible set. The generation rule \eqref{RuleForAS} can be simplified to \eqref{NewRuleForAS}.
\end{proof}

Proposition \ref{SimplifyPropo} indicates that all of the non-zero commutators of the operators in a former accessible set $G_{m-1}$ and the operators in $\bar{\digamma}$ should be added into the accessible set $G_m$. Compared with \eqref{RuleForAS}, \eqref{NewRuleForAS} avoids a full search of the elements in $\Omega$. Using \eqref{RuleForAS}, the average computation complexity of finding a single element in the set $\Omega$ is $O(3^N2^{3N})$. Using \eqref{NewRuleForAS}, the computational complexity of updating an element is reduced to $O(2^{3N})$.

Problem \ref{initialProblem} can now be restated as the following problem with a lower computational complexity.
\begin{problem}\label{SimplifyProblem}
	Develop an economic method to generate the accessible set $G=f(\Omega,\bar{\digamma}, \bar{M})$ with a good ordering, using rules \eqref{NewRuleForAS01} and \eqref{NewRuleForAS}.
\end{problem}

%To put it differently, the graph generated by the accessible set $G$ is connected.

\iffalse
one usage of Lemma \eqref{bisearch} is to simplify the measurement operator set.
\begin{example}\label{exmSimplifyMset}
	Suppose we have the set $M=\{X_1Y_2, X_1X_2\}$ and the set $\bar{\digamma}=\{Z_1,Z_2,Z_3,\quad \cdot\quad,Z_N\}$. We have 
	\begin{equation}
	[X_1Y_2,Z_2]=2\mi X_1X_2,
	\end{equation}
	which means $X_1Y_2$ and $X_1X_2$ can generate each other using the rule \eqref{NewRuleForAS}. Thus, according to Lemma \eqref{bisearch}, we can eliminate one operator from $M$ and the set can be reduced to $M^1=\{X_1Y_2\}$ or $M^2=\{X_1X_2\}$. Both $M^1$ and $M^2$ generate the same accessible set.
\end{example}
For the above example, the workload of generating an accessible set has been reduced to a half by eliminating one operator in $M$. Actually, eliminating operators in the measurement set can effectively avoid the repetition of searching.
\fi

%%%=====================================================================================
\subsection{Graphs generated by accessible sets}
Graphs can be employed to demonstrate the generation of accessible sets. We benefit from graphs mainly in three aspects. First, a graph visualizes the relationship between operators in the corresponding accessible set. Moreover, the repetition pattern revealed by a graph when generating an accessible set has the potential to be summarized and used to extend an accessible set to any given qubit number. Second, graphs can be used to arrange the ordering of element operators in the state vector to achieve a good structure of the state space matrices. Third, graphs can help with the proofs of our lemmas and propositions. 

We assign each accessible set $G$ a graph $\mathbb{G}$. The vertices of $\mathbb{G}$ are elements in the corresponding accessible set $G$. There is an edge $\langle O_m, O_n\rangle$ between two vertices $O_m$ and $O_n$ if and only if there exists a $\nu \in \bar{\digamma}$  such that
\begin{equation}
\text{Tr}(O_n^\dagger\lfloor O_m,\nu\rceil)\neq 0.
\end{equation}
We use such a $\nu$ to label the edge $\langle O_m, O_n\rangle$ and $\nu$ is called the \textit{edging operator}. The graph $\mathbb{G}$ can be described as $\mathbb{G}=\{G,\mathbb{E}\}$, where the accessible set $G$ is a set of vertex operators and $\mathbb{E}$ is the set of all of the edges. Moreover, we have the following definition.
\begin{definition}\label{edging}
	A path in the graph can be specified by a set of vertex operators $(O_1,O_2,\cdots, O_m)$ or by the starting operator, ending operator and a sequence of edging operators $\{O_1,(\nu_1,\nu_2, \cdots), O_m\}$. We refer to the sequence $E=(\nu_1,\nu_2, \cdots)$ as an \textit{edging sequence} which is a sequence of edging operators. $S(\bar{\digamma})$ is the set of all of the sequences of finite elements of edging operators chosen from $\bar{\digamma}$. Then, the notation $E\in S(\bar{\digamma})$ indicates that all of the elements in $E$ belong to $\bar{\digamma}$. We define $C(E)=\underbrace{\nu_1,\nu_2, \cdots}$ as the collection of edging operators in $E$.
\end{definition}
\begin{remark}
	The reason that the triplet $\{O_1,E, O_m\}$ can specify a path is based on the fact that the graphs in this paper are all simple graphs. It is worth noting the differences between a set, a collection and a sequence. Sets and sequences can be regarded as specific classes collections that are endowed with different features. While the uniqueness of objects in a collection is not guaranteed, a set is defined as a collection of distinct objects. While objects in a collection may not be ordered, elements in a sequence are uniquely ordered. For example, While $E_1=(X,Y,Y)$ and $E_2=(Y,X,Y)$ are two different sequences, the collections $C_E^1=C(E_1)=\underbrace{X,Y,Y}$ and $C_E^2=C(E_2)=\underbrace{Y,X,Y}$ are the same. Moreover, we have $E_1,E_2\in S(\{X,Y\})$ which indicates that sets of edging operators forming the sequences $E_1$ and $E_2$ are the same.
\end{remark}

Labeling the vertices of graph $\mathbb{G}$ with natural numbers, we obtain the adjacency matrix $\mathbb{A}$ whose $(i,j)$-th entry is 1, if and only if there is an edge connecting the $i$-th and $j$-th vertices \cite{Wilson96graph}. The state space matrix $A$ in \eqref{statespaceeq} has the same structure as $\mathbb{A}$, while having different elements from $\mathbb{A}$. The graph and the matrix $A$ share the same pattern in a certain sense.

 Based on the fact that $\lfloor\sigma_m,\sigma_i\rceil=\sigma_n$ and $\lfloor\sigma_m,\sigma_j\rceil=\sigma_n$ yield $\sigma_i=\sigma_j$ where $i,j,m,n \in \{0,1,2,3\}$, we have 
\begin{equation}\label{nomultiedge01}
\nu=u \quad  \nu,\ u\in \bar{\digamma}
\end{equation}
if
\begin{equation}\label{nomultiedge02}
\begin{split}
\begin{cases}
&\text{Tr}(O_n^\dagger\lfloor O_m,\nu\rceil)\neq 0,\\
&\text{Tr}(O_n^\dagger\lfloor O_m,u\rceil)\neq 0, 
\end{cases}\quad
\text{for some}\  O_m,  O_n\in \Omega.
\end{split}
\end{equation}
Hence, there are no multiple edges with the same direction between any two vertices which means the labeling of every edge is unique. Also, note that we always have 
\begin{equation}
\text{Tr}(O_m^\dagger\lfloor O_m,\nu\rceil)= 0
\end{equation}
for any $O_m,\nu\in\Omega$. This means there exists no edge $\langle O_m,O_m\rangle$ and therefore there is no loop in the graph.
We conclude that all of the graphs associated with accessible sets defined in this paper have no loops or multiple edges, which means they are simple graphs.

A graph is called undirected if there is no direction assigned to the edges. We have the following lemma which states that all of the graphs generated by accessible sets are essentially undirected:
\begin{lemma}\label{undirectedgraph}
	Assume that $\mathbb{G}=\{G,\mathbb{E}\}$ where $G=f(\Omega,\bar{\digamma},\bar{M})$ and $\mathbb{E}$ is the corresponding set of edges. Then each edge of $\mathbb{G}$ is bi-directed if endowed with direction. 
\end{lemma}
\begin{proof}
	Suppose $O_m$ and $O_n$ are two different vertices and there is an edge $\langle O_m, O_n \rangle$ connecting $O_m$ and $O_n$. We prove that there exists an edge $\langle O_n, O_m \rangle$ and it has the same label as $\langle O_m, O_n \rangle$.
	
	According to the definition of an edge and the fact that we have an edge $\langle O_m, O_n \rangle$, there exists a $\nu \in \bar{\digamma}$ such that
	\begin{equation}
	O_n=\lfloor O_m,\nu\rceil.
	\end{equation}
	Then the edge $\langle O_m, O_n \rangle$ is labeled by $\nu$.
	According to \eqref{PauliCommutaRelation}, we have
	\begin{equation}
	O_m=\lfloor O_n,\nu\rceil.
	\end{equation}
	Then the edge $\langle O_n, O_m \rangle$ is also labeled by $\nu$. Since the edges $\langle O_m, O_n \rangle$ and $\langle O_n, O_m \rangle$ share the same vertices and label, the pair of vertices $O_m$ and $O_n$ are unordered. Since all of the edges are undirected, the graph is undirected. To put it differently, the iterative rules given in \eqref{iterativeAS} and \eqref{NewRuleForAS} can achieve a bi-directional search. 
\end{proof}
Considering Lemma \ref{undirectedgraph}, direction becomes a trivial property for graphs representing accessible sets. Hence, we regard all graphs employed  in this paper to be undirected. 

 Note that, a graph is connected if there exists at least one path between every pair of vertices. An \textit{induced subgraph} of a graph is another graph, formed from a subset of the vertices of the graph and all of the edges connecting pairs of vertices in that subset. We have the following lemma.
\begin{lemma}\label{connectedgraph01}
	Let $\mathbb{G}=\{G,\mathbb{E}\}$ where $G=f(\Omega,\bar{\digamma},\bar{M})$ and $\mathbb{E}$ is the corresponding set of edges. Also let $\mathbb{M}=\{M,\mathbb{E}_M\}$ be an induced graph of $\mathbb{G}$ where all of the vertices of $\mathbb{M}$ are in the measurement set $M$ and $\mathbb{E}_M$ is the corresponding set of edges. If the graph $\mathbb{M}$ is connected, then the graph $\mathbb{G}$ is connected.
\end{lemma}
\begin{proof}
	All of the elements in the accessible set $G$ are generated by the elements in the initial set $M$. Thus, they are connected with the elements in $M$ according to the definition of the graph $\mathbb{G}$. Since $\mathbb{M}$ is assumed to be connected, the graph $\mathbb{G}$ is also connected.
\end{proof}

\begin{lemma}\label{connectedgraph02}
	Given $G=f(\Omega,\bar{\digamma},\bar{M})$ and $\tilde{G}=f(\Omega,\bar{\digamma},\tilde{M})$ where $\tilde{M}$ is a non-empty subset of $G$. If $G$ is connected, we have $\tilde{G}=G$.
\end{lemma}
\begin{proof}
Since we suppose that an undirected graph $G$ is connected, then the accessible set can be obtained starting from an arbitrary group of operators (not necessary the measurement operators) that belong to the accessible set, using the generation rules \eqref{NewRuleForAS01} and \eqref{NewRuleForAS}. Then Lemma \ref{connectedgraph02} follows. \end{proof}

%%==============================spin chain systems=============================

\subsection{Special consideration for a class of spin chain systems}
\begin{figure}	
	\centering		
	\includegraphics[width=8cm]{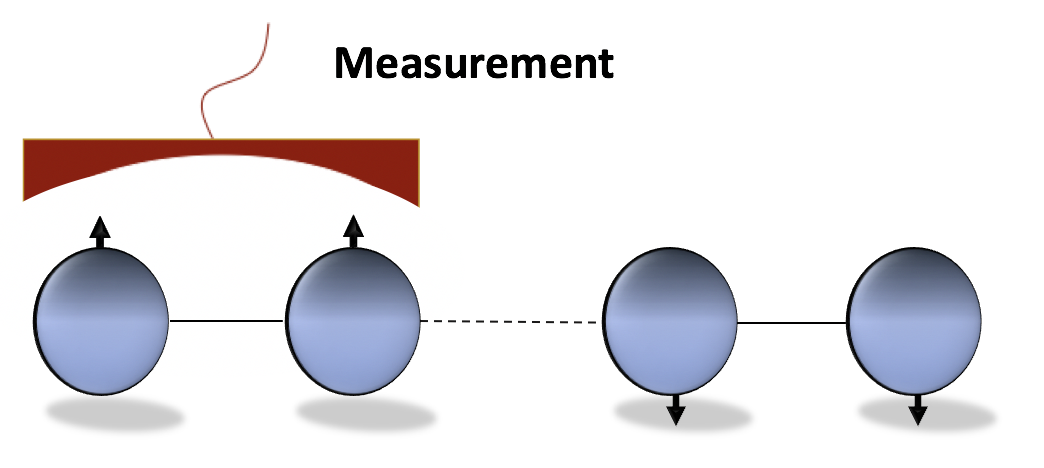}		
	\caption{An example of a quantum  network system whose elements are qubit systems coupled in the form of a chain. The measurement is on the first several (two in this example) qubits in the chain system.}
	\label{QuantumNetwork}		
\end{figure}
A chain system, where qubits are connected in the form of a string, is a fundamental and typical quantum network system (see Fig. \ref{QuantumNetwork})\cite{yuanlong2018Algorithm, Akira2017Hamiltonian}. Here, we consider a chain system consisting of $N$ qubits \cite{Akira2017Hamiltonian,yuanlong2018Algorithm}. The system Hamiltonian is 
\begin{equation}\label{hamiltonianmodel}
H=\sum_{k=1}^{N-1} h_k (X_k X_{k+1} + Y_{k}Y_{k+1})
\end{equation}
where the following notation is used $X\coloneqq \sigma_x$, $Y\coloneqq \sigma_y$ and $Z\coloneqq\sigma_z$. The subscript $k$ indicates that the operator is on the $k$-th qubit. The operator $X_k X_{k+1}$ represents 
\begin{equation}
I^{\otimes (k-1)}\otimes X_k\otimes X_{k+1}\otimes I^{\otimes (N-k-1)}.
\end{equation}
To write the operators in a compact form, we omit the tensor product symbol and the identity operator unless otherwise specified. The system whose Hamiltonian is given in \eqref{hamiltonianmodel} is an exchange model without transverse field \cite{Di2008,Chris2005}. The coupling Hamiltonian between the $k$-th and $(k+1)$-th qubit is $h_k (X_k X_{k+1} + Y_{k}Y_{k+1})$. 
The decomposed set $\bar{\digamma}$ for the chain system in \eqref{hamiltonianmodel} is
\begin{equation}\label{ExchangeDigamma}
\bar{\digamma} = \{ X_1 X_2 ,\   Y_1 Y_2,\  \cdots, \   X_kX_{k+1},\  Y_kY_{k+1},\  \cdots,  \} 
\end{equation}
where $1\leq k\leq N-1$.

For quantum chain systems, we present the following proposition to help with the generation of accessible sets for the system with Hamiltonian given in \eqref{hamiltonianmodel}.
\begin{proposition}\label{ExmSimpleClaim} 
	Given $\bar{\digamma}$ as in \eqref{ExchangeDigamma} and the measurement set $\bar{M}=\{Z^{\otimes (m-1)} X_m\}$, we have
	\begin{equation}\label{GXClaim1}
	G^X\coloneqq f(\Omega,\bar{\digamma},\bar{M})=\{O_1^X,\  \cdots,\  O_k^X, \ \cdots \}
	\end{equation}
	where
	\begin{equation}
	O_k^X=
	\begin{cases}
	Z^{\otimes (m+k-1)} X_{m+k}, \quad \text{$k$ is even},\\
	Z^{\otimes (m+k-1)} Y_{m+k}, \quad \text{$k$ is odd},
	\end{cases}
	\end{equation}
	and $0 \leq k\leq N-m$.	Similarly, if the measurement set is given as $\bar{M}=\{Z^{\otimes (m-1)} Y_m\}$, the corresponding accessible set is 
	\begin{equation}\label{GYClaim1}
	G^Y\coloneqq f(\Omega,\bar{\digamma},\bar{M})=\{O_1^Y,\  \cdots,\  O_k^Y, \ \cdots\  \}
	\end{equation}
	where
	\begin{equation}
	O_k^Y=
	\begin{cases}
	Z^{\otimes (m+k-1)} Y_{m+k}, \quad \text{$k$ is even},\\
	Z^{\otimes (m+k-1)} X_{m+k}, \quad \text{$k$ is odd} ,
	\end{cases}
	\end{equation}
	and $0 \leq k\leq N-m$.
\end{proposition}

\begin{proof}
	%	We have the following identity for the calculation of commutators.
	%	\begin{equation}
	%	[AB,CD]=A[B,C]D+[A,C]BD+CA[B,D]+C[A,D]B,
	%	\end{equation}
	%	where $AB=A\otimes B$ and $CD=C\otimes D$. 
	According to \eqref{NewRuleForAS}, the iterative generation rule involves adding non-zero operators that are generated by taking the commutator operation on operators in $G_{m-1}$ and operators in $\bar{\digamma}$ into the new accessible set $G_{m}$. Here we find the following common patterns
	\begin{equation}\label{rule1}
	\begin{split}
	&[O_{(1,k-1)} Z_kY_{k+1},X_{k+1}X_{k+2}] \\
	&\ \ \ \ =O_{(1,k-1)} Z_k[Y_{k+1},X_{k+1}]X_{k+2} \\
	&\ \ \ \ =(-2\mi)O_{(1,k-1)} Z_kZ_{k+1}X_{k+2}  ; \\
	&[O_{(1,k-1)} Z_kX_{k+1},Y_{k+1}Y_{k+2}]\\
	&\ \ \ \ =O_{(1,k-1)} Z_k[X_{k+1},Y_{k+1}]Y_{k+2} \\
	&\ \ \ \ =(2\mi)O_{(1,k-1)} Z_kZ_{k+1}Y_{k+2}  ; \\
	&[O_{(1,k-1)} Z_kY_{k+1},Y_{k+1}Y_{k+2}]\\
	&\ \ \ \ =O_{(1,k-1)} Z_k[Y_{k+1},Y_{k+1}]X_{k+2}=0; \\
	&[O_{(1,k-1)} Z_kX_{k+1},X_{k+1}X_{k+2}]\\
	&\ \ \ \ =O_{(1,k-1)} Z_k[X_{k+1},X_{k+1}]Y_{k+2}=0 ,\\
	\end{split}
	\end{equation}
	where $O_{(1,k-1)}$ is an operator acting on the first $(k-1)$ operators. For a system whose Hamiltonian takes the form of \eqref{hamiltonianmodel}, $O_{(1,k-1)}Z_kY_{k+1}\in G$ where $1\leq k\leq N-2$ leads to $O_{(1,k-1)}Z_kZ_{k+1}X_{k+2}\in G$. If we have $O_{(1,k-1)}Z_kX_{k+1}\in G$ where $1\leq k\leq N-2$, then we also have $O_{(1,k-1)}Z_kZ_{k+1}Y_{k+2}\in G$. The equalities in \eqref{rule1} provide us with operators that should be added when all of the operators in $G$ can be written as either in the form of $O_{(1,k-1)} Z_kY_{k+1}$ or in the form of $O_{(1,k-1)} Z_kX_{k+1}$. Note that, the added operators $O_{(1,k-1)}Z_kZ_{k+1}X_{k+2}$ and $O_{(1,k-1)}Z_kZ_{k+1}Y_{k+2}$ can be written in the form $O_{(1,k)}Z_{k+1}X_{k+2}$ and $O_{(1,k)}Z_{k+1}Y_{k+2}$, which facilitates the iterative generation of accessible sets.
\end{proof}
Proposition \ref{ExmSimpleClaim} provides us with accessible sets for cases such as (a), (c) and (e) in Section \ref{Sec4}. 

%%======================== improving the ordering =========================================

\subsection{Improving the ordering}\label{SecOrdering}
The results in Section \ref{SecComplex} concern the reduction of computational complexity. Here, we focus on the generation of accessible sets with good ordering. Two main objectives are:
\begin{itemize}
	\item To find a repetition pattern for the state-space model as the number of nodes increases;
	\item To reveal the connections between element operators in $G$.
\end{itemize}
These two objectives are vital for finding a repetition pattern for the state space model and writing down an $N$-qubit system model for arbitrary $N$. Otherwise, one only has accessible sets for several limited values of $N$, and the identification, analysis and control of the system will be difficult to be extended. Arranging the order of element operators in the state vector according to the graph, it is likely to obtain a state space model with good structure. 

\begin{definition}\label{BasisOperator}
	We denote the set $\mathcal{B}=\{I_{2\times2},\sigma_x.\sigma_y,\sigma_z\}$ as the cell set and an operator $O\in\mathcal{B}$ is a cell operator. 
\end{definition}
In this paper, we use the notation $X\coloneqq\sigma_x$, $Y\coloneqq \sigma_y$ and $Z\coloneqq \sigma_z$ interchangeably so the cell set can also be written as $\mathcal{B}=\{I_{2\times 2},X,Y,Z\}$.
%In this paper, we use the notation $X\def \sigma_x$ interchangeably so the cell set can also be written as $\mathcal{B}=\{I_{2\times 2},X,Y,Z\}$.

\begin{definition}\label{kfinite}
	A set $G$ is said to be \textit{$k$-finite} if every operator $O\in G$ takes the following form
	\begin{equation}\label{kthqubitoperator01}
	O=\sigma_{s_1^k}\otimes \sigma_{s_2^k}\otimes \sigma_{s_3^k}\otimes\cdots\otimes \sigma_{s_j^k}\cdots
	\end{equation}
	where
	\begin{equation}\label{kthqubitoperator001} 
	s_j^k\in
	\begin{cases}
	\{0,1,2,3\}, \quad 1\leq j<k,\\
	\{1,2,3\}, \qquad j=k,\\
	\{0\},\quad\qquad k<j\leq M.
	\end{cases}
	\end{equation}
	Here, $M<\infty$ is the number of cell operators that form operators in $O$.
\end{definition}

We start from an $N$-qubit chain system with a Hamiltonian as in \eqref{hamiltonianmodel}. For such a system, we have the following proposition:
\begin{proposition}\label{SubsetsLemma}
	For an $i$-qubit network system with the Hamiltonian given in \eqref{hamiltonianmodel}, $\bar{\digamma}_i$ given in \eqref{ExchangeDigamma} and $\bar{M}$ connected, let $G_i=f(\Omega, \bar{\digamma}_i, \bar{M})$ be the corresponding accessible set. Define a series of sets $G_{\lfloor k}$ where $1\leq k\leq i$ 
   \begin{equation}\label{seperate}
   G_{\lfloor k}=
   \begin{cases}
   G_k, \quad k=1,\\
   G_k-G_{k-1}, \quad k>1.
   \end{cases}
   \end{equation}
 The operator $``-"$ acting on any sets $A$ and $B$ denoted by $A-B$ indicates the subtraction of the set $B$ from the set $A$. We have the following assertions.
 
   	Assertion 1: For $\forall 1\leq l\leq \mu \leq i$, we have $G_l \subseteq G_\mu$.
   
   Assertion 2: The set $G_{\lfloor k}$ is $k$-finite.
   
   Assertion 3: There exists $\bar{\digamma}_k\subset \bar{\digamma}$ such that $G_{\lfloor k}=f(\Omega,\bar{\digamma}_k,\{O_k\})$ where $O_k$ can be any operator in $G_{\lfloor k}$ and $\bar{\digamma}_k$ can be independent of the choice of $O_k$. 
\end{proposition}
Proposition \ref{SubsetsLemma} reveals the relations between the sets $G_{k}$ and $G_{\lfloor k}$ for $k=1,2,\cdots$.  Please see Appendix \ref{APP1} for proof.

\iffalse %-----------------------------following deleted---------------------------------------------
\begin{proposition}\label{SubsetsLemma01}
	For an $N$-qubit network system with the Hamiltonian given in \eqref{hamiltonianmodel} and $\bar{\digamma}$ given in \eqref{ExchangeDigamma}, denote $G=f(\Omega, \bar{\digamma}, \bar{M})$ and the graph $\mathbb{G}=\{G_N, \bar{M}\}$ is connected. Define $G_{\lfloor k}$ as a subset of $G_N$ whose elements take the following form
		\begin{equation}\label{kthqubitoperator00001}
	\sigma_{i_1^k}\otimes \sigma_{i_2^k}\otimes \sigma_{i_3^k}\otimes\cdots\otimes \sigma_{i_j^k}\otimes\cdots\otimes \sigma_{i_N^k}
	\end{equation}	
	where
	\begin{equation}\label{kthqubitoperator0001}
	i_j^k\in
	\begin{cases}
	\{0,1,2,3\}, \quad j<k,\\
	\{1,2,3\}, \quad j=k,\\
	\{0\},\quad j>k.
	\end{cases}
	\end{equation}
	Then $G_N$ can be divided as disjoint unions of $G_k$
	\begin{equation}\label{seperate01}
	G_N=\dot\cup_{k=1}^N G_{\lfloor k}.
	\end{equation}

Assertion 1: For a system of $i$ qubits where $0<i\leq N$, the accessible set is
\begin{equation}\label{lessqubit}
G_i=\cup_{k=1}^i G_{\lfloor k}
\end{equation}
where $G_{\lfloor k}$ is defined in \eqref{seperate}.

Assertion 2: There exists $\bar{\digamma}_k\subset \bar{\digamma}$ such that $G_{\lfloor k}=f(\Omega,\bar{\digamma}_k,\{O_k\})$ where $O_k$ is an arbitrary operator in $G_{\lfloor k}$. 
\end{proposition}
\fi %------------------------------------above deleted-----------------------------------------

Equation \eqref{seperate} is equivalent to $G_{i}=G_{k-1}\cup G_{\lfloor k}$, which means one only needs to find $G_{\lfloor k}$ to obtain the accessible set $G_i$ given the accessible set $G_{k-1}$ for a class of spin chain systems. Moreover, if we observe a pattern shared by all of the graphs $\mathbb{G}_{\lfloor k}$, one can generate the accessible set $G_n$ for any given $n$. Furthermore, Assertion 3 in Proposition \ref{SubsetsLemma} confirms that all of the induced subgraphs $\mathbb{G}_{\lfloor k}$ are connected. The connectivity of $\mathbb{G}_{\lfloor k}$ indicates that all of the subsets $G_{\lfloor k}$ can be generated by starting from an arbitrary operator that belongs to $G_{\lfloor k}$. After finding an arbitrary operator $O\in G_{\lfloor k}$, one can obtain all of the operators in $G_{\lfloor k}$.

 We want to design a search algorithm that is suitable for generating all of the subsets $G_{\lfloor k}$. In the set $G$, we place the elements of $G_{\lfloor k}$ in front of the elements of $G_{\lfloor k+1}$. For different systems and measurement schemes, one needs to design a proper search rule accordingly. The main idea employed in generating an accessible set with a good ordering is to divide the accessible set $G$ into subsets to reveal a generation pattern that is shared by the accessible sets as the number of qubits increases.

Here, we summarize the generation process. Given a measurement scheme, we first decompose the measurement set and the Hamiltonian set into the form we defined in Definition \ref{decomposition}. Then we observe the measurement set to see if Proposition \ref{ExmSimpleClaim} can be applied to this situation. For some cases, we can obtain an accessible set at this stage. Otherwise, we determine if the graph associated with the accessible set is connected or not. If the graph is connected, we divide the accessible set into subsets to find certain repetition pattern when generating the subsets. If the graph associated with an accessible set is not connected, this paper can still provide some insight. Generally, a graph can be divided into several connected sub-graphs. The ideas in this paper can thus still be applied for the generation of the connected subgraphs. Collecting all of the vertices of the subgraphs together provides a complete accessible set. 
%以上这个就是我们总结出来的置换规则，可以利于手算的。尝试一下能不能把这个规则总结得更好一点。
%Note： 在example 的时候要不要放上 state-space model呢？？？问问董老师
%===================      algorithm     =========================================================

\section{Illustrative examples}\label{Sec4}

Here we present several examples to demonstrate the generation of a proper accessible set with good ordering. The object system is a chain system consisting of $N$ qubits. The system Hamiltonian is given in \eqref{hamiltonianmodel} and the set $\bar{\digamma}$ is given in \eqref{ExchangeDigamma}. We provide accessible sets for the following six measurement schemes:

\begin{inparaenum}
	\item [(a)] $M=\{X_1\}$;\ \ \ 
	\item [(b)] $M=\{Z_1\}$;\qquad\ 
	\item [(c)] $M=\{Z_1Y_2\}$;
	
	\item [(d)] $M=\{Y_1 Z_2\}$;
	\item [(e)] $M=\{Z_1Z_2 X_3\}$;
	\item [(f)] $M=\{X_1 Y_2 Z_3\}$.
\end{inparaenum}

For cases (a) and (b), only the first qubit in the chain system is measured. For cases (c) and (d), we measure the first two qubits of the chain system. For cases (e) and (f), the first three qubits are measured. These cases cover most of the common fundamental measurement settings, and several similar settings are omitted. For example, from the analysis on case (a) one can straightforwardly write down the analysis result when the measurement is $M=\{Y_1\}$.

To visualize the generation process, we employ graphs to describe accessible sets. According to Proposition \ref{SubsetsLemma}, when $\bar{M}$ has only one element, the graph $\mathbb{G}$ generated by a complete accessible set $G$ is connected, which means there is always a path connecting any two operators in $G$. This holds for all of the examples in this section and is clearly exemplified by case (b) (See Fig. \ref{M1Z1}). The graph $\mathbb{G}_{\lfloor k}$ associated with subset $G_{\lfloor k}$ is also connected under the assumption in Proposition \ref{SubsetsLemma}. This can also be observed from all of the examples, especially from cases (b), (d) and (f). 

 For cases (a), (c) and (e), we present analytical formula for the accessible set for an arbitrary number $N$. For cases (b), (d) and (f), we present the generation of the accessible for a fixed qubit number $N$, employing graphs to find the repetition pattern generating the accessible set. By observing and summarizing those generation patterns, we can determine the accessible set for any given $N$. Arranging the elements according to the graphs can provide us with a good structure for the state space equation matrices $A$, $B$ and $C$ in \eqref{statespaceeq}.

%====================================== One ==========================================
\paragraph{\textbf{Measuring $\mathbf{X_1}$}}
According to Proposition \ref{ExmSimpleClaim}, the accessible set $G$ can be obtained immediately as
\begin{equation}
G=\begin{cases}
\{X_1,Z_1Y_2, Z_1Z_2X_3,\ \cdots,\   Z^{\otimes (N-1)}Y_N\},\  \text{$N$ is even},\\
\{X_1,Z_1Y_2, Z_1Z_2X_3,\ \cdots,\   Z^{\otimes (N-1)}X_N\},\  \text{$N$ is odd}.
\end{cases}
\end{equation}

\paragraph{\textbf{Measuring $\mathbf{Z_1}$}}
 We have the following iterative generation rule:
\begin{equation}\label{zktozk+1}
\begin{cases}
\lfloor Z_k, X_kX_{k+1}\rceil=Y_kX_{k+1},\\
\lfloor Y_kX_{k+1},Y_{k}Y_{k+1}\rceil=Z_{k+1}.
\end{cases}
\end{equation}
From \eqref{zktozk+1} and the fact that $Z_1$ is in the accessible set, it can be identified that the operators $Z_k$, where $1\leq k \leq N$, are all in the accessible set $G$.

Aiming to find all of the other operators in the accessible set, we divide the accessible set $G$ into the following subsets
\begin{equation}
G=\cup_{k=1}^N G_{\lfloor k}=\cup_{k=1}^N \{Z_k,\  \cdots\ \}
\end{equation}
where the subset $G_{\lfloor k}$ is $k$-finite. 

We denote $Z_k$ as the `core' operator in the subset $G_{\lfloor k}$. A `core' operator is an operator selected from $G_{\lfloor k}$ and serves as the starting operator while generating $G_{\lfloor k}$. Since the graph $G_{\lfloor k}$ is connected, one can select any operator in $G_{\lfloor k}$ to be a core operator according to Proposition \ref{SubsetsLemma}, which means that all of the other operators in $G_{\lfloor k}$ can be generated from $Z_k$ by rule \eqref{NewRuleForAS}. 

In Fig. \ref{M1Z1}, the accessible set is given for the case where there are six qubits in the network system. Starting from the core operator, the generation of the operators in $G$ forms a graph which follows a clear repetition pattern. In subset $G_{\lfloor 1}$ (in the blue dashed box), there is only one operator $Z_1$ which is the measurement operator. In subset $G_{\lfloor 2}$ (in the yellow dashed box), there are three operators $Z_2$, $X_1Y_2$ and $Y_1X_2$. Following the special pattern revealed in Fig. \ref{M1Z1}, one can generate an accessible set for a chain system with an arbitrary number of qubits. Moreover, we can also turn to Fig. \ref{M1Z1} for a good ordering when constructing the system state variable $\boldsymbol{x}$.

%====================================== Two ==========================================

\paragraph {\textbf{Measuring $\mathbf{Z_1Y_2}$}}
\begin{figure}	
	\centering		
	\includegraphics[width=9cm]{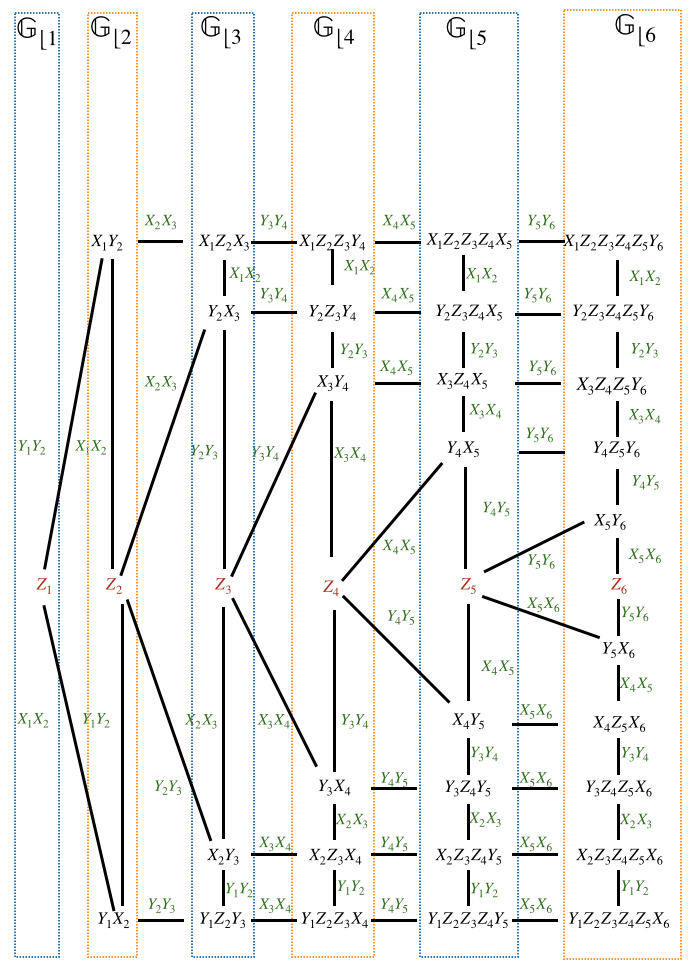}		
	\caption{Accessible set $G$ when measuring $Z_1$. The system contains six qubits. Element operators (marked in black and red) are vertices of the graph. Edges connecting vertices are labeled by operators (marked in green) in the set $\bar{\digamma}$ used to generate the vertices.}
	\label{M1Z1}		
\end{figure}
Given the initial measurement operator $Z_1 Y_2$, the accessible set is as follows according to Proposition \ref{ExmSimpleClaim}
\begin{equation}
G=\{O_1,O_2, \ \cdots\ , O_k, \ \cdots\ \},
\end{equation}
where
\begin{equation}
O_k=
\begin{cases}
Z_1Z_2\cdots Z_{k-1}Y_k, \quad \text{$k$ is even},\\
Z_1Z_2 \cdots Z_{k-1}X_k, \quad \text{$k$ is odd} .
\end{cases}
\end{equation}
\paragraph{\textbf{Measuring $\mathbf{Y_1 Z_2}$}}
We have the following equality
\begin{equation}
\begin{cases}
\lfloor Y_1Z_2,Y_1Y_2\rceil=X_2,\\
\lfloor X_2,Y_2Y_3\rceil=Z_2Y_3,
\end{cases}
\end{equation}
which indicates that the operator $Z_2 Y_3$ is in the accessible set $G$. Therefore, from Proposition \ref{ExmSimpleClaim}, the following operators are in the accessible set $G$
\begin{equation}
\{X_2,\ Z_2 Y_3,\ Z_2 Z_3X_4, \ \cdots, \  I\otimes Z^{\otimes (k-2)}O_k,\ \cdots \  \}\subset G,
\end{equation}
where $3\leq k \leq N$ and 
\begin{equation}\label{Appy1z2}
O_k=
\begin{cases}
X\quad  \text{$k$ is even},\\
Y\quad  \text{$k$ is odd}.
\end{cases}
\end{equation}
Aiming to find all of the other operators in $G$, we divide it into the following subsets
\begin{equation}
\begin{split}
G =& G_{\lfloor 2} \cup G_{\lfloor 3} \cup G_{\lfloor 4} \cup\  \cdots\  \cup G_{\lfloor k}\  \cdots \\
=& \{X_2,\ Y_1Z_2\} \cup \{Z_2 Y_3,\  \cdots\ \} \cup \{Z_2 Z_3X_4,\  \cdots\ \} \cup\cdots 
\end{split}
\end{equation}
where the subset $G_{\lfloor k}$ is $k$-finite. 

Let the `core' operator of $G_{\lfloor k}$ be $O_k^c$ 
\begin{equation}
O_k^c=
\begin{cases}
Z_2 \cdots Z_{k-1}X_k, \  \text{$k$ is even},\\
Z_2 \cdots Z_{k-1}Y_k, \  \text{$k$ is odd} .
\end{cases}
\end{equation}
According to Proposition \ref{SubsetsLemma}, all of the other operators in $G_{\lfloor k}$ can be generated from $O_k^c$ by rule \eqref{NewRuleForAS} given that the core operator $O_k^c$ belongs to the subset $G_{\lfloor k}$. 
\begin{figure}	
	\centering		
	\includegraphics[width=7cm]{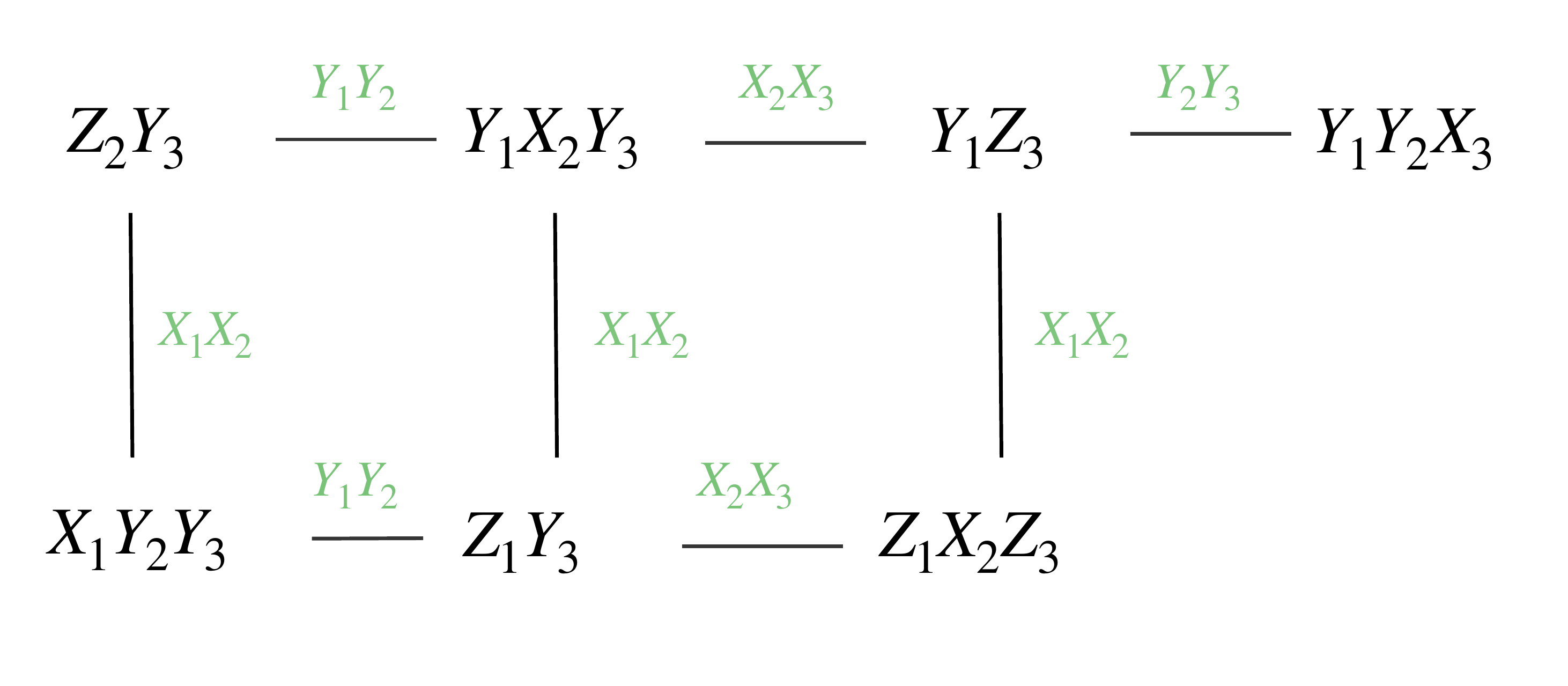}		
	\caption{The generation procedure for operators that form the set $G_{\lfloor 3} $.}
	\label{GenerationRulesN3}		
\end{figure}
\begin{figure}	
	\centering		
	\includegraphics[width=8cm]{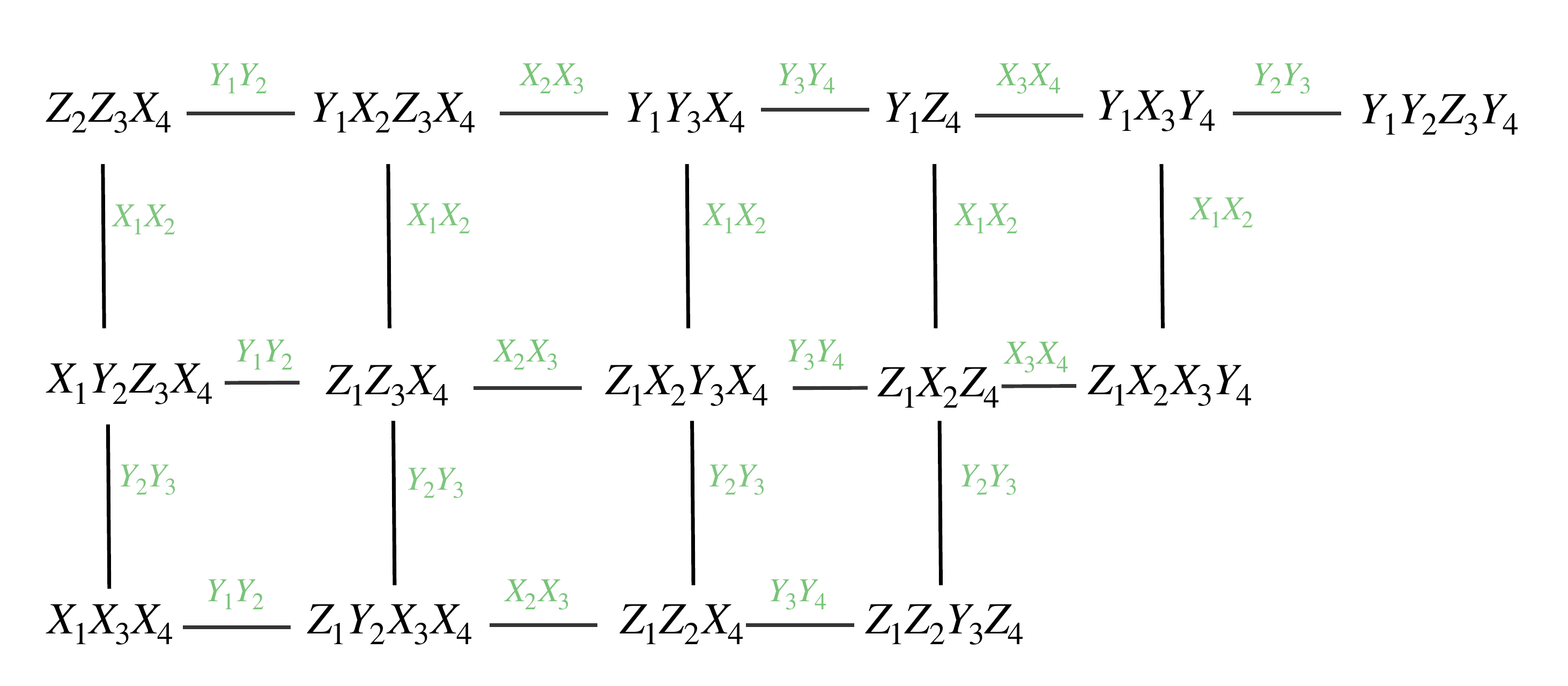}		
	\caption{The generation procedure for operators forming the set $G_{\lfloor 4} $. The operators in black are in the accessible set while operators in green are in $\bar{\digamma}$.}
	\label{GenerationRulesN4}		
\end{figure}
\begin{figure}	
	\centering		
	\includegraphics[width=8.5cm]{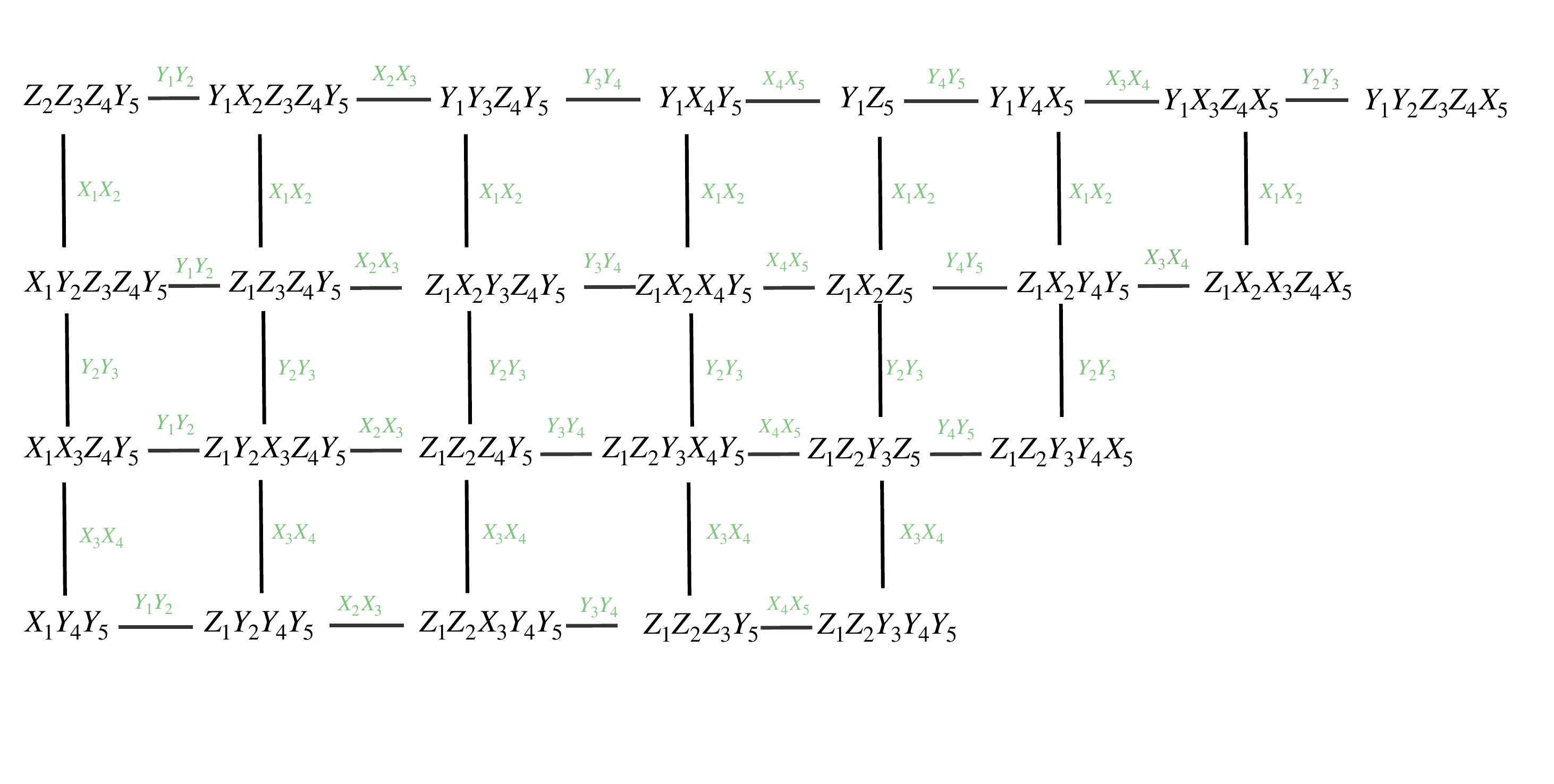}		
	\caption{The generation procedure for operators in set $G_{\lfloor 5} $.}
	\label{GenerationRulesN5}		
\end{figure}

The subset $G_{\lfloor 2} =\{ Y_1 Z_2,\  X_2  \}$ only contains two operators. Starting from the core operator $Z_2 Y_3$, elements in $G_{\lfloor 3}$ can be inferred and the generation procedure is shown in Fig. \ref{GenerationRulesN3}. For example, given that $Z_2 Z_3X_4\in G_{\lfloor 4}$ and $[Z_2 Z_3X_4,Y_1 Y_2]=-2\mi Y_1 X_2 Z_3X_4$, it can be inferred that $Y_1 X_2 Z_3X_4\in G_{\lfloor 4}$ as well, according to \eqref{NewRuleForAS}. Similarly, the generating processes and element operators for the subsets $G_{\lfloor 4}$ and $G_{\lfloor 5}$ areshown in Fig. \ref{GenerationRulesN4} and Fig. \ref{GenerationRulesN5}, respectively. The generation patterns for those sets are similar and repetitive. It can be seen that the number of elements in $G_{\lfloor 3}$ is $3+4=7$; the number of elements in $G_{\lfloor 4}$ is $4+5+6=15$; the number of elements in $G_{\lfloor 5}$ is $5+6+7+8=26$. Using the induction method, the number of operators in $G_{\lfloor k}$ is $\frac{(3k-2)(k-1)}{2}$.

Let the total number of qubits in the chain system be denoted as $N$. The total number of operators in the accessible set $G$ is
\begin{equation}
|G|=\sum_{K=2}^{N}\frac{(3k-2)(k-1)}{2}  = \frac{N^3-N^2}{2}.
\end{equation}
From the analysis above, it is clear that the number of operators in $G$ scales as $N^3$, which can be far more than the qubit number.

%====================================== Three ==========================================
\paragraph{\textbf{Measuring $\mathbf{Z_1Z_2 X_3}$}}
	Given the initial measurement operator $Z_1Z_2 X_3$, the accessible set is as follows according to Proposition \ref{ExmSimpleClaim}
	\begin{equation}
	G=\{O_1,\ O_2, \  \cdots,\  O_k,\  \cdots \ \},
	\end{equation}
	where
	\begin{equation}
	O_k=
	\begin{cases}
	Z_1Z_2\cdots Z_{k-1}Y_k, \  \text{$k$ is even},\\
	Z_1Z_2 \cdots Z_{k-1}X_k, \  \text{$k$ is odd} .
	\end{cases}
	\end{equation}
\paragraph{\textbf{Measuring $\mathbf{X_1 Y_2 Z_3}$}}	
\begin{figure}	
	\includegraphics[width=0.6\textwidth,center]{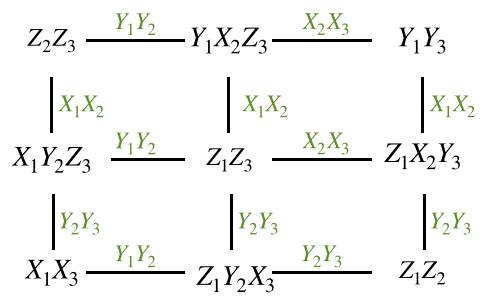}		
	\caption{The accessible set when measuring $X_1 Y_2 Z_3$. The number of qubits is $3$.}
	\label{3sensorN3}		
\end{figure}
\begin{figure}		
	\includegraphics[width=0.8\textwidth,center]{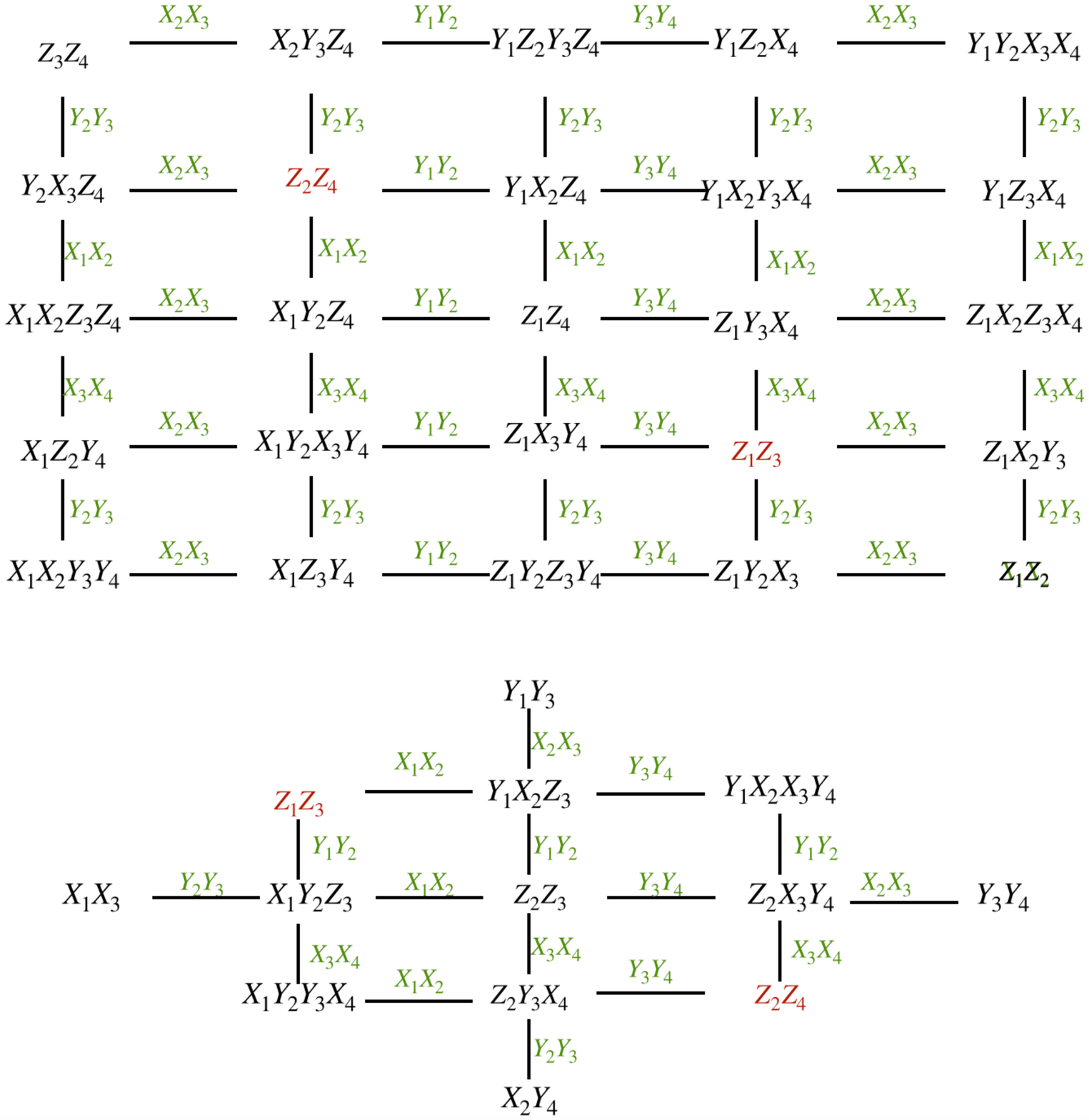}	
	\caption{The accessible set when measuring $X_1 Y_2 Z_3$. The number of qubits is $4$. The graph generated by $G_4$ can not be presented in an uncrossed two dimensional graph. Hence, we separate it into two graphs. However, there are repetitive elements in the two graphs (see the ones marked in red), which means the complete graph is still a connected simple graph. }
	\label{3sensorN4}		
\end{figure}
	When the measurement operator is $X_1 Y_2 Z_3$, the accessible set $G$ for $3$-qubit and $4$-qubit systems are shown in Fig. \ref{3sensorN3} and Fig. \ref{3sensorN4}, respectively. In both situations, the graphs show some complicated generation patterns with a certain repetition mode.
%===========================================================================================================================
\section{Conclusion}\label{Conclusion}
The modeling of a quantum network system is one of the basic tasks for many problems such as quantum system identification, quantum control, quantum sensing and quantum filtering. In this paper, we investigated the problem of modeling a class of quantum network systems as the state space model, which is widely used in quantum engineering. To develop the state space models, a major task is to obtain an accessible set. We mainly focus on the generation of accessible sets given a system Hamiltonian and a measurement operator. We obtained a series of results that can simplify the generation procedure for accessible sets for a class of network systems. We also employed graphs to demonstrate the generation of accessible sets and to guide the ordering of elements in the state space vectors. Several examples were presented where the accessible sets for different measurement schemes were obtained. 

%+=========================== Appendix ========================================================
%\appendix

\begin{appendix}[Proof of Proposition \ref{SubsetsLemma}]

\section{Appendix: Proof of Proposition \ref{SubsetsLemma} }\label{APP1}
In order to prove Proposition \ref{SubsetsLemma}, we first present preliminaries and several lemmas that will be used.

To simplify the narrative, we divide elements in $G_{\lfloor k+1}$ into two classes $G^1_{\lfloor k+1}$ and $G^2_{\lfloor k+1}$ such that:
\begin{itemize}
	\item[1.] $G^1_{\lfloor k+1}\subset G_{\lfloor k+1}$ is the set of operators that are adjacent to elements in $G_k$;
	\item[2.] $G^2_{\lfloor k+1}\subset G_{\lfloor k+1}$ is the set of operators that are not adjacent to any element in $G_k$.
 \end{itemize}

 Note that the following three statements are equivalent: 1) $G^1_{\lfloor k+1}$ can be generated by the triplet $\{\Omega,\bar{\digamma}_k,\{O_{k+1}\}\}$ where $O_{k+1}$ is an arbitrary operator that belongs to $G^1_{\lfloor k+1}$; 2) The graph $\mathbb{G}^1_{\lfloor k+1}$ is connected; 3) There is always a path between every pair of vertices in the graph $\mathbb{G}^1_{\lfloor k+1}$. In a connected graph, there are no unreachable vertices. 
 
Noticing that an arbitrary operator in $\Omega$ is formed by the tensor product on a sequence of cell operators in $\mathcal{B}=\{I_{2\times2},X,Y,Z\}$. For a given operator $O\in \Omega$, we refer to the cell operator on the $j$th qubit as the $j$th cell operator. Moreover, for an operator $O\in G_k$ where $G_k$ is $k$-finite, we refer to the $k$th cell operator as the \textbf{\textit{ending operator}}. For example, the third cell operator of $O=X\otimes Y\otimes Z\otimes I$ is $Z$ and the ending operator of $O$ is $I$, given that $O\in G_{4}$. We can also say that the operator $O$ ends with $I$. Now we give two definitions and two lemmas.

 \begin{definition}\label{adjacentVertex}
 	Given $G=f(\Omega,\bar{\digamma},\bar{M})$ and its corresponding graph $\mathbb{G}$, two vertices $O_m\in G$ and $O_n\in G$ are called adjacent if and only if there exists $\nu \in \bar{\digamma}$ such that
 	\begin{equation}
 	\text{Tr}(O_n^\dagger,[O_m,\nu])\neq 0,
 	\end{equation}
 	which means there is an edge (labeled by $\nu$) connecting vertices $O_m$ and $O_n$ in the graph $\mathbb{G}$.
 \end{definition}
\begin{definition}\label{adjacentGraph}
	Graph $\mathbb{G}_m$ and graph $\mathbb{G}_n$ are adjacent if and only if there exists a vertex in $\mathbb{G}_m$ that is adjacent to a vertex in $\mathbb{G}_n$.
\end{definition}

\begin{lemma}\label{Ass2&3G101}
	 Suppose the system Hamiltonian is given in \eqref{hamiltonianmodel}, $\bar{\digamma}_i$ is given in \eqref{ExchangeDigamma} and $\bar{M}$ is connected. For the vertex $O_a\in G_{\lfloor k}$, if there exist $O_c, O_d\in G^1_{\lfloor k+1}$ which are adjacent to $O_a$, we have the following statements:
	\begin{itemize}
		\item[] 1) $O_c$ and $O_d$ are $(k+1)$-finite.
		\item[] 2) $O_c$ and $O_d$ are connected.
	\end{itemize}
\end{lemma}
\begin{proof}
We follow the same notation as given in Proposition \ref{SubsetsLemma}. For the system given in Lemma \ref{Ass2&3G101}, the only case where $O_a$ has two adjacent operators in $G^1_{\lfloor k+1}$ is when $O_a$ ends with $Z$, where two edging operators $X_kX_{k+1},Y_kY_{k+1}$ can be applied to $O_a$ and generating operators in $G^1_{\lfloor k+1}$. Fig. \ref{Triangle} demonstrates the case where two edges $Y_kY_{k+1}$ and $X_kX_{k+1}$ leading the operator $O_{(1,k-1)}Z_k$ in $G_{\lfloor k}$ to operators $O_{(1,k-1)}X_kY_{k+1}$ and $O_{(1,k-1)}Y_kX_{k+1}$ in $G^1_{\lfloor k+1}$. Here, $O_{(1,k-1)}$ represents an arbitrary operator on the first $(k-1)$ qubits. We can find a path $(O_{(1,k-1)}X_kY_{k+1}, O_{(1,k-1)}I_kZ_{k+1}, O_{(1,k-1)}Y_kX_{k+1})$, where $O_{(1,k-1)}I_kZ_{k+1}$ is also in $G^1_{\lfloor k+1}$. Thus, operators generated in this pattern are connected and are $(k+1)$-finite. 
\end{proof}

\begin{figure}
	\includegraphics[width=.55\textwidth]{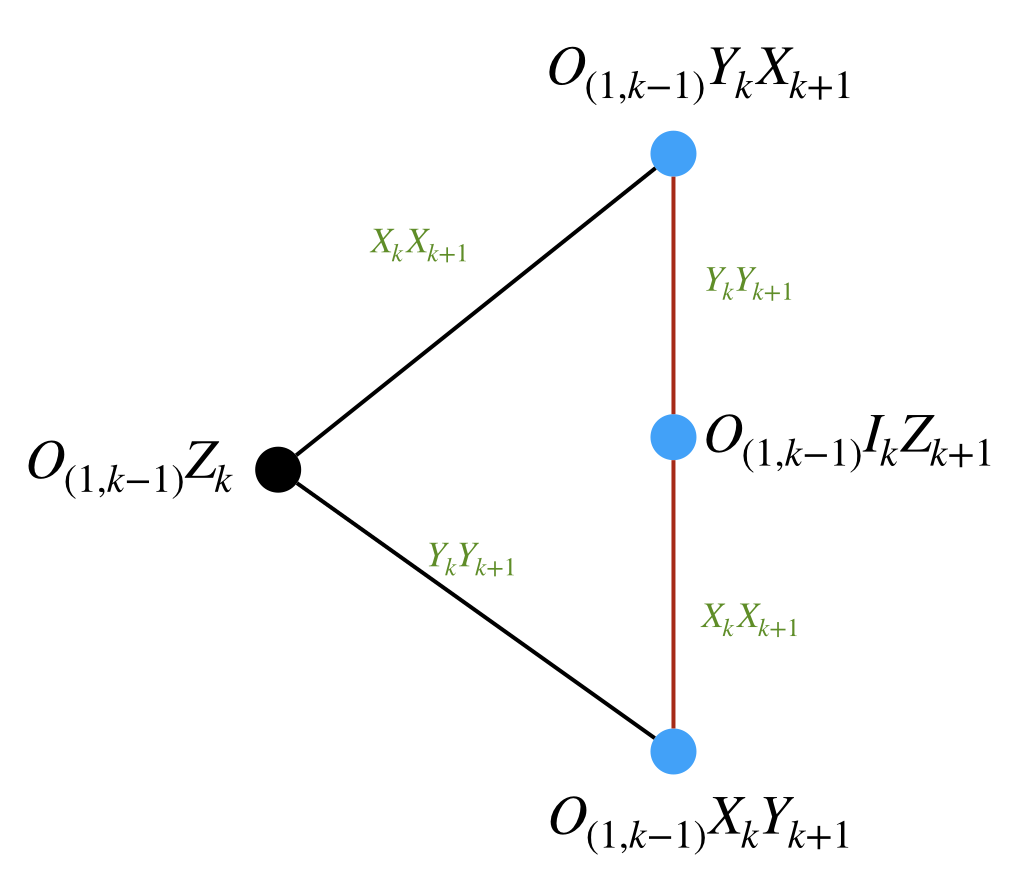}	
	\caption{The operator in $G_{\lfloor k}$ ends with cell operator $Z$. The edges with different labels connecting it to two operators in $G^1_{\lfloor k+1}$.}	
	\label{Triangle} 
\end{figure} 

\begin{lemma}\label{Ass2&3G1}
 Suppose the system Hamiltonian is as given in \eqref{hamiltonianmodel}, $\bar{\digamma}_i$ is as given in \eqref{ExchangeDigamma} and $\bar{M}$ is connected. Two vertices $O_a, O_b\in G_{\lfloor k}$ are adjacent. If there exist $O_c\in G^1_{\lfloor k+1}$ which is adjacent to $O_a$, we have the following statements:
\begin{itemize}
\item[] 1) There exists $O_d\in G^1_{\lfloor k+1}$ which is adjacent to $O_b$. 
\item[] 2) $O_c$ and $O_d$ are $(k+1)$-finite.
\item[] 3) There exists a path in $G_{\lfloor k+1}$ that connects $O_c$ and $O_d$.
\end{itemize}
\end{lemma}
\begin{proof}
We follow the same notation as given in Proposition \ref{SubsetsLemma}.
Considering the ending cell operator, the elements in $G_{\lfloor k}$ can be classified into three classes: operators whose ending operator is $X$; operators whose ending operator is $Y$; operators whose ending operator is $Z$. Thus there are $6$ possible pairs of vertices which are adjacent:
\begin{itemize}
\item[] 1) $O^\alpha_{(1,k-1)}X_k$ and $O^\beta_{1,k-1}X_k$;
\item[] 2) $O^\alpha_{(1,k-1)}X_k$ and $O^\beta_{1,k-1}Y_k$;
\item[] 3) $O^\alpha_{(1,k-1)}X_k$ and $O^\beta_{1,k-1}Z_k$;
\item[] 4) $O^\alpha_{(1,k-1)}Z_k$ and $O^\beta_{1,k-1}Z_k$;
\item[] 5) $O^\alpha_{(1,k-1)}Y_k$ and $O^\beta_{1,k-1}Y_k$;
\item[] 6) $O^\alpha_{(1,k-1)}Y_k$ and $O^\beta_{1,k-1}Z_k$.
\end{itemize}
Note that case 5) is similar to case 1) and case 6) is similar to case 3). We only consider case 1) to 4). Moreover, we state that case 2) does not exist in our graph. This can be proved by contradiction. Suppose $O^\alpha_{(1,k-1)}X_k,O^\beta_{1,k-1}Y_k\in G_{\lfloor k}$ are adjacent. Then there exist $E\in \bar{\digamma}_k$ such that
\begin{equation*}
	\lfloor O^\alpha_{(1,k-1)}X_k, E\rceil=O^\beta_{1,k-1}Y_k.
\end{equation*}
Rewriting $E$ in the form $O^E_{(1,k-1)}O^E_k$, we have 
\begin{equation} \label{App001}
\lfloor O^\alpha_{(1,k-1)}X_k, O^E_{(1,k-1)}O^E_k\rceil=O^\beta_{1,k-1}Y_k.
\end{equation}
Equations \eqref{OperationDecomp} and \eqref{App001} indicate that $O^E_k=Z$ which is not possible since $E\in \bar{\digamma}_k$ where $\bar{\digamma}_k=\{X_{k-1}X_k,Y_{k-1}Y_k\}$. Thus, operators $O^\alpha_{(1,k-1)}X_k$ and $O^\beta_{1,k-1}Y_k$ can not be adjacent. Thus, to prove Lemma \ref{Ass2&3G1}, it suffices to prove that operators generated by the above pairs of vertices in cases 1), 3) and 4) are connected and $(k+1)$-finite, given that the vertices in each pair are adjacent.

 For operators in $G_{\lfloor k]}$ with ending operator $X$, the edge leading it to $G^1_{\lfloor k+1}$ can only be $Y_{k}Y_{k+1}$. For operators whose ending operator is $Y$, the edge can only be $X_{k}X_{k+1}$. For operators whose ending operator is $Z$, the edges can be chosen to be either $X_{k}X_{k+1}$ or $Y_{k}Y_{k+1}$.

\begin{figure*}[]
	\centering
	\subfigure[]{\includegraphics[width=0.23\textwidth]{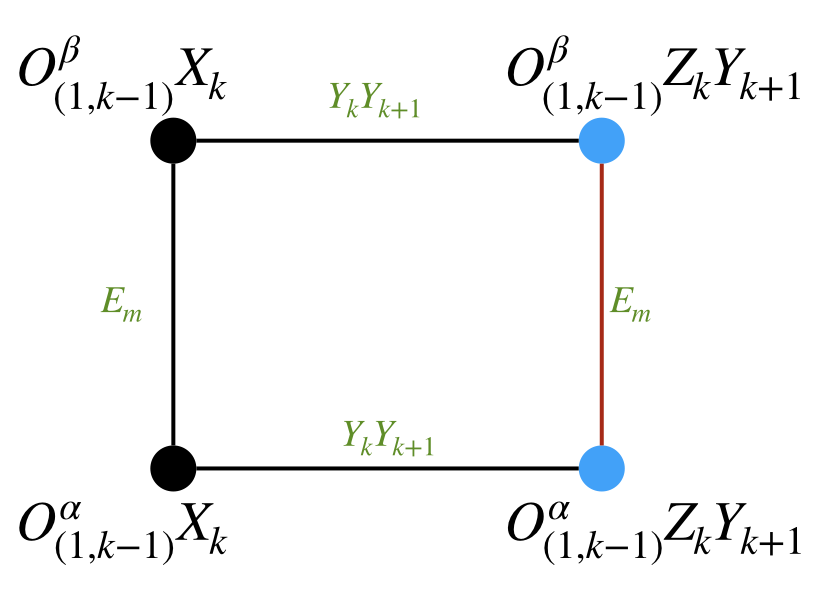}\label{Smallsquare02}}\hspace{0.0\textwidth}
	\subfigure[] {\includegraphics[width=0.3\textwidth]{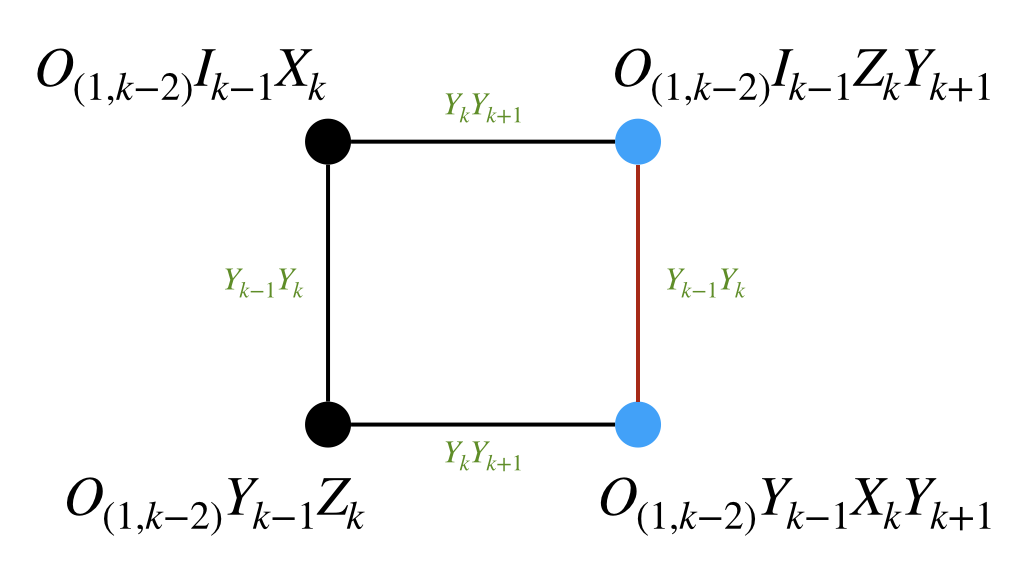}\label{smallsquare}}
	\subfigure[] {\includegraphics[width=0.22\textwidth]{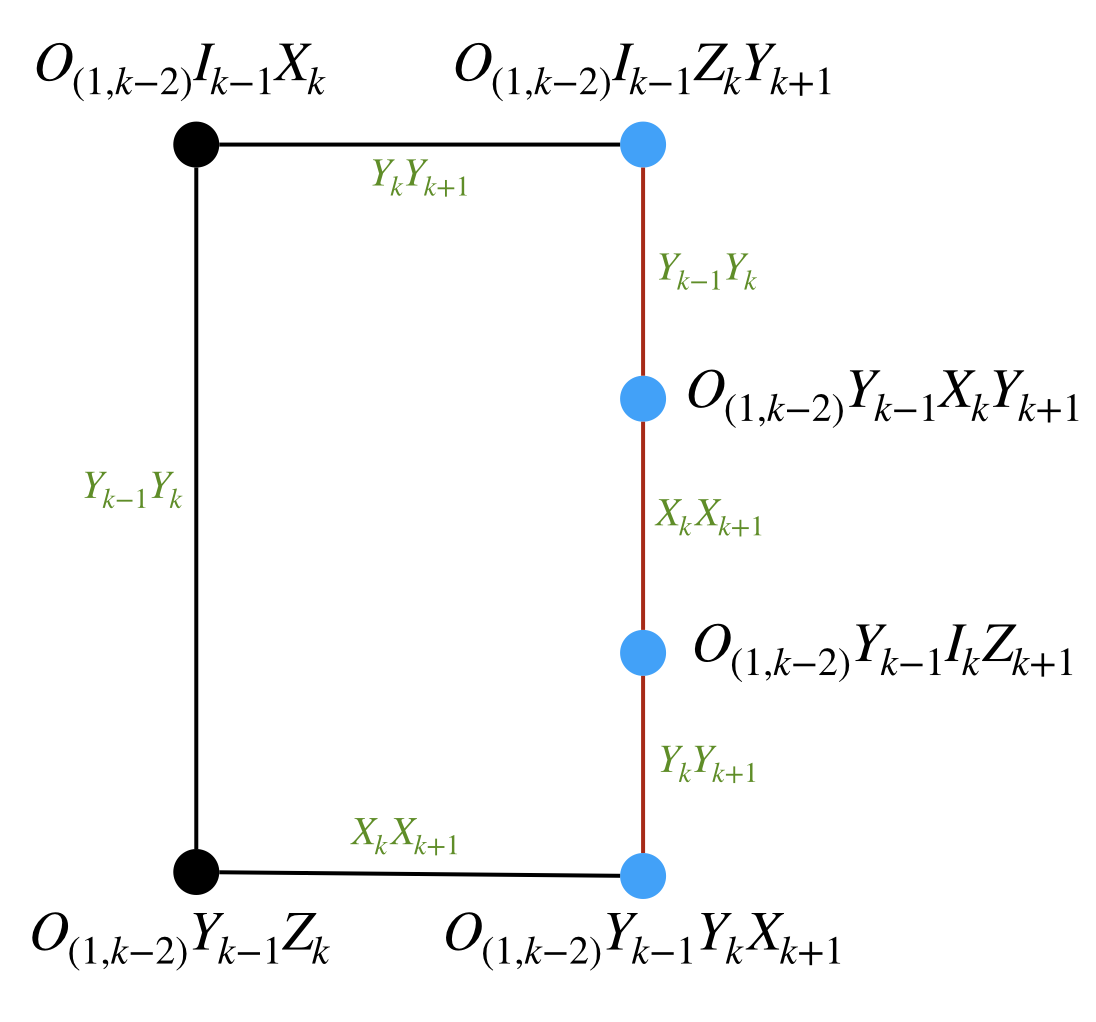}\label{LargeSquare}} \hspace{0.0\textwidth}
	\subfigure []{ \includegraphics[width=0.2\textwidth]{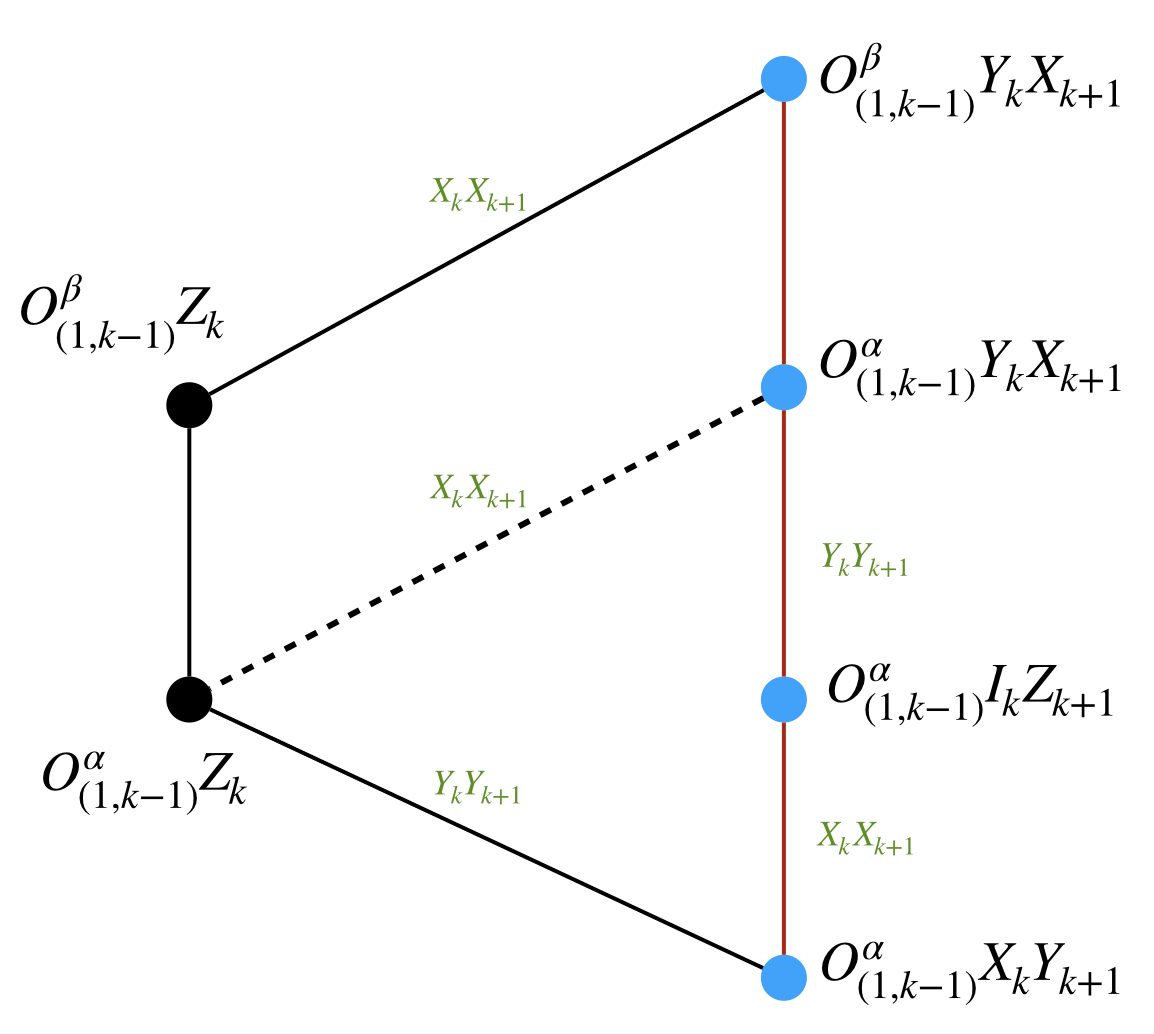} \label{decomposed}	}
	\caption{Four generation patterns connecting $G_{\lfloor k}$ with $G^1_{\lfloor k+1}$. The black vertices represent operators in $G_{\lfloor k}$ and blue vertices represent operators in $G^1_{\lfloor k+1}$. For the four patterns, we show that all blue vertices are connected within each pattern, which proves Lemma \ref{Ass2&3G1}.}
	\label{Gpatterns}
\end{figure*}
%	\caption{Four generation patterns connecting $G_{\lfloor k}$ with $G^1_{\lfloor k+1}$. Sub-figure \ref{Smallsquare02} shows the generation pattern $1$: Two operators in $G_{\lfloor k}$ have the same ending operator $X_k$. The edges connecting them with the operators in $G^1_{\lfloor k+1}$ have the same label $Y_kY_{k+1}$. Sub-figure \ref{smallsquare} shows the generation pattern $2$: Two operators in $G_{\lfloor k}$ have different cell operators at $k$th position. One operator ends with $X_k$ and the other ends with $Z_k$. The edges connecting them with the operators in $G^1_{\lfloor k+1}$ have the same labeling operator $X_k$. The edges connecting them with the operators in $G^1_{\lfloor k+1}$ have the same label $Y_kY_{k+1}$. Sub-figure \ref{LargeSquare} shows the generation pattern $3$: Two operators in $G_{\lfloor k}$ have different ending operators $X_k$ and $Z_k$. The edges connecting them to operators in $G^1_{\lfloor k+1}$ are different. Sub-figure \ref{decomposed} shows the generation pattern $4$: A special case which can be decomposed as a combination of the case in Figure \ref{decomposed} and the case in Figure \ref{Smallsquare02}. }

% make it clear what is the criterion to do the classification and why this four classes can represent all of the cases.
In Fig. \ref{Gpatterns}, we present the corresponding operators in $G^1_{\lfloor k+1}$ that can be generated by the four classes of adjacent vertices in $G_{\lfloor k}$. Note that for the case $3)$ where the adjacent vertices are $O^\alpha_{(1,k-1)}X_k$ and $O^\beta_{1,k-1}Z_k$, we further divide the case into two sub-cases: the edging operators that connecting operators in $G_{\lfloor k}$ and $G^1_{\lfloor k+1}$ are the same; the edging operators that connecting operators in $G_{\lfloor k}$ and $G^1_{\lfloor k+1}$ are different. Thus there are five different generation patterns as shown in Fig. \ref{Gpatterns}.

For pattern 1, see Fig. \ref{Smallsquare02}. Two operators $O^\alpha_{(1,k-1)}X_k$ and $O^\beta_{(1,k-1)}X_k$ are in the set $G_{\lfloor k}$ and are connected by the edging operator $\nu_m$. Here, $O^\alpha_{(1,k-1)}$ and $O^\beta_{(1,k-1)}$ represent two different operators on the first $k-1$ qubits. In $G^\alpha_{\lfloor k+1}$, we have two operators $O^\alpha_{(1,k-1)}Z_kY_{k+1}$ and $O^\beta_{(1,k-1)}Z_kY_{k+1}$ that are generated by operators $O^\alpha_{(1,k-1)}X_k$ and $O^\beta_{(1,k-1)}X_k$ in $G_{\lfloor k}$, respectively. Both edges are labeled by $Y_kY_{k+1}$. Since $\nu_m$ is an operator on two adjacent qubits within the first $k-1$ qubits while $Y_kY_{k+1}$ is on the $k$th and $(k+1)$th systems, $\nu_m$ and $Y_kY_{k+1}$ commute to each other. We prove that the generated operators $O^\alpha_{(1,k-1)}Z_kY_{k+1}$ and $O^\beta_{(1,k-1)}Z_kY_{k+1}$ are also connected by the edge labeled by $\nu_m$. Since $[O^\alpha_{(1,k-1)}X_k,\nu_m]=O^\beta_{(1,k-1)}X_k$, we have $[O^\alpha_{(1,k-1)},\nu_m]=O^\beta_{(1,k-1)}$ which yields that
\begin{equation}
\begin{split}
[O^\alpha_{(1,k-1)}Z_kY_{k+1},\nu_m]&=[O^\alpha_{(1,k-1)},\nu_m]\otimes Z_kY_{k+1}\\
&=O^\beta_{(1,k-1)}Z_kY_{k+1}.
\end{split}
\end{equation}
Thus, operators generated in pattern 1 are connected and $(k+1)$-finite. 

For pattern 2, see Fig. \ref{smallsquare}. Two operators $O^\alpha_{(1,k-2)}I_{k-1}X_k$ and $O^\beta_{(1,k-2)}Y_{k-1}Z_k$ are in the set $G_{\lfloor k}$ and are connected by the edge $Y_{k-1}Y_k$. Similar to pattern 1, the two generated operators $O^\alpha_{(1,k-2)}I_{k-1}Z_kY_{k+1}$ and $O^\beta_{(1,k-2)}Y_{k-1}X_kY_{k+1}$ in $G^1_{\lfloor k+1}$ are connected by the edge $Y_{k-1}Y_k$. Thus, operators generated in pattern 2 are connected and $(k+1)$-finite. 

For pattern 3, see Fig. \ref{LargeSquare}. Two operators in $G_{\lfloor k}$ have different ending operators $X_k$ and $Z_k$. Edges connecting them to operators in $G^1_{\lfloor k+1}$ have different labels $X_kX_{k+1}$ and $Y_kY_{k+1}$. For this case, we can always find a path that connects the two generated operators and all vertices in the path are in $G^1_{\lfloor k+1}$. In the case presented in Fig. \ref{LargeSquare}, the path is $O_{(1,k-2)}I_{k-1}Z_kY_{k+1}\rightarrow O_{(1,k-2)}Y_{k-1}X_kY_{k+1}\rightarrow O_{(1,k-2)}Y_{k-1}I_kZ_{k+1}\rightarrow O_{(1,k-2)}Y_{k-1}Y_kX_{k+1}$. Thus, operators generated in pattern 3 are connected and $(k+1)$-finite.

For pattern 4, see Fig. \ref{decomposed}. The two operators in $G_{\lfloor k}$ end with the $Z$ operator while edges leading them to operators in $G^1_{\lfloor k+1}$ are different. This pattern can be decomposed into a combination of pattern 1 and pattern 2. Thus the two generated operators in this pattern are connected and $(k+1)$-finite.

The above analysis shows that if two operators in $G_{\lfloor k+1}$ are adjacent to two operators in $G_{\lfloor k}$ which are adjacent, the two operators in $G_{\lfloor k+1}$ are connected. Moreover, there is a path in $G_{\lfloor k+1}$ that connects the two operators in $G_{\lfloor k+1}$.
\end{proof}

To facilitate the proof process, we decompose the operation $\lfloor\cdot\rceil$ on operators in $\Omega$ to a series of operations on the cell operators in $\mathcal{B}$. We have
\begin{equation}\label{APPlemma04}
\begin{split}
[O_a\otimes O_d,&O_b\otimes O_e]=\\
&[(O_a\otimes I)(I\otimes O_d),(O_b\otimes I)(I\otimes O_e)]\\
=&(O_a\otimes I)[(I\otimes O_d),(O_b\otimes I)](I\otimes O_e)\\
&+[(O_a\otimes I),(O_b\otimes I)](I\otimes O_d)(I\otimes O_e)\\
&+(O_b\otimes I)(O_a\otimes I)[(I\otimes O_d),(I\otimes O_e)]\\
&+(O_b\otimes I)[(O_a\otimes I),(I\otimes O_d)](I\otimes O_e)\\
=&[(O_a\otimes I),(O_b\otimes I)](I\otimes O_d)(I\otimes O_e)\\
&+(O_b\otimes I)(O_a\otimes I)[(I\otimes O_d),(I\otimes O_e)]\\
=&[O_a,O_b]\otimes O_dO_e+O_bO_a\otimes[O_d,O_e].
\end{split}
\end{equation}
Since $O_a, O_d\in\Omega$, which means they are tensor products of Pauli matrices, we have
\begin{equation*}
\begin{split}
&[\sigma_a,\sigma_b]\otimes \sigma_d\sigma_e+\sigma_b\sigma_a\otimes[\sigma_d,\sigma_e]=\\
&\begin{cases}
2\mi \varepsilon_{def} \sigma_b\sigma_a\otimes \sigma_f, & (a=0 \veebar b=0)\wedge (d\neq e \wedge d\neq 0 \wedge e\neq 0);\\
2\mi \varepsilon_{abc} \sigma_c\otimes \sigma_d\sigma_e, & (d=0 \veebar e=0)\wedge (a\neq b \wedge a\neq 0 \wedge b\neq 0);\\
2\mi \varepsilon_{def} I \otimes \sigma_f, & (a=b)\wedge (d\neq e);\\
2\mi \varepsilon_{abc}  \sigma_c \otimes I, & (a\neq b)\wedge (d= e);\\
0 & \text{otherwise}.
\end{cases}	
\end{split}
\end{equation*}
Here, $\varepsilon$ is the Levi-Civita symbol, $\wedge$ is the logical conjunction symbol and $ \veebar$ is the exclusive disjunction symbol. Then we have
\begin{equation}\label{OperationDecomp}
\begin{split}
&\lfloor \sigma_a\otimes \sigma_d,\sigma_b\otimes \sigma_e \rceil=\\
&\begin{cases}
\sigma_b\boxdot\sigma_a\otimes \lfloor \sigma_d,\sigma_e \rceil, &\\
 \qquad   \qquad \qquad \quad (a=0 \veebar b=0)\wedge (d\neq e \wedge d\neq 0 \wedge e\neq 0);& \\
\lfloor \sigma_a,\sigma_b \rceil\otimes \sigma_d\boxdot\sigma_e, &\\
 \qquad  \qquad \qquad \quad (d=0 \veebar e=0)\wedge (a\neq b \wedge a\neq 0 \wedge b\neq 0);& \\
\sigma_a\boxdot\sigma_b\otimes  \lfloor \sigma_d,\sigma_e \rceil, \qquad (a=b)\wedge (d\neq e);\\
\lfloor \sigma_a,\sigma_b \rceil \otimes \sigma_d\boxdot\sigma_e, \qquad (a\neq b)\wedge (d= e);\\
0 \qquad \qquad \text{otherwise}.
\end{cases}	
\end{split}
\end{equation}
The operator $\boxdot$ is defined as $A\boxdot B=O$ such that $O\in \Omega$ and $O\propto AB$. 
Equation \eqref{OperationDecomp} indicates that the operation $\lfloor \cdot,\cdot \rceil$ on $O_a\in \Omega$ and $O_b\in \Omega$ can be decomposed into first $\lfloor \cdot,\cdot \rceil$ and $\boxdot$ on cell operators in $\mathcal{B}$ and then $\otimes$ of the results. Then we find that the operations $\lfloor \cdot,\cdot \rceil$ and $\boxdot$ obtained by the decomposition only act on cell operators of the same position. Thus, some properties that apply on both operations $\lfloor \cdot,\cdot \rceil$ and $\boxdot$ on cell operators can also be generalized to the operation $\lfloor \cdot,\cdot \rceil$ on operators in $\Omega$.

Define $\bar{\digamma}_{\lfloor k}=\bar{\digamma}_{k}-\bar{\digamma}_{k-1}$. For the system with Hamiltonian given in \eqref{ExchangeDigamma}, we have $\bar{\digamma}_{\lfloor k+1}=\{X_kX_{k+1},Y_kY_{k+1}\}$. Since operators in $\bar{\digamma}_{\lfloor k+1}$ can only relate operators in $G_{\lfloor k}$ and $G^1_{\lfloor k+1}$, all of the operators in $G^1_{\lfloor k+1}$ are generated by operators in $G_{\lfloor k}$. Thus, to prove Assertion 2 and Assertion 3, it suffices to consider elements in $G_{\lfloor k}$. We assert that every element in $G_{\lfloor k}$ can generate at most two elements in $G^1_{\lfloor k+1}$ for a system whose Hamiltonian takes the form of \eqref{hamiltonianmodel}. For the case where an operator in $G_{\lfloor k}$ can generate two operators in $G^1_{\lfloor k+1}$, we present Lemma \ref{Ass2&3G101} to confirm that the generated operators in $G^1_{\lfloor k+1}$ are connected and $(k+1)$-finite. For the case where an operator in $G_{\lfloor k}$ only generates one operator in $G^1_{\lfloor k+1}$, we present Lemma \ref{Ass2&3G1} to confirm that the generated operators in $G^1_{\lfloor k+1}$ are connected and $(k+1)$-finite.

%From the above analysis, we show that there is a path connecting all of the elements in $G^1_{\lfloor k+1}$ that generated by pattern 1 to 4. Moreover, the vertices on the path are all operators that are $(k+1)$-finite. 

%%=============second class=======================================

The idea for the proof of the following lemmas is to interpret the lemmas with regard to the cell operators and then the conclusion can be generalized to operators in $\Omega$. Fig. \ref{PauliCommute} and Fig. \ref{PauliMultip} demonstrate all of the possible cases of applying the operator $\lfloor \cdot,\cdot \rceil$ and $\boxdot$ to operators in $\mathcal{B}$. Note that the operator $I$ is not in the graph in Fig. \ref{PauliCommute} since applying $\lfloor \cdot,\cdot \rceil$ to $I$ always results in a zero matrix. It can be seen that starting from an arbitrary vertex in the graph, there is a path leading to any operator in $\{X, Y, Z\}$. It can be seen from Fig. \ref{PauliCommute} that applying the operator $\lfloor \cdot,\cdot \rceil$ to one Pauli operator with the same edging operator for an even time results in the original Pauli operator. Thus, one can reach the same ending vertex by removing an even number of the same edging operators from the path. A similar conclusion can be obtained from Fig. \ref{PauliMultip} for the operation $\boxdot$. Starting from an arbitrary vertex, one can reach any operator of interest in this graph and one can reach the same ending vertex by removing an even number of the same edging operators in $\{X, Y, Z\}$ from the path, even if these edges are not adjacent. For instance, starting from $Z$ at the right corner, the ending operators are $Y$ for the path $\{Z,(I, Y, X, Y),Y\}$, $\{Z,(I, X),Y\}$ and the path $\{Z,(I, Z, X, Z),Y\}$. Now we provide three lemmas that will be used for the proof of Assertion 2 of Proposition \ref{SubsetsLemma}.
\begin{figure}		
	\includegraphics[width=0.35\textwidth,center]{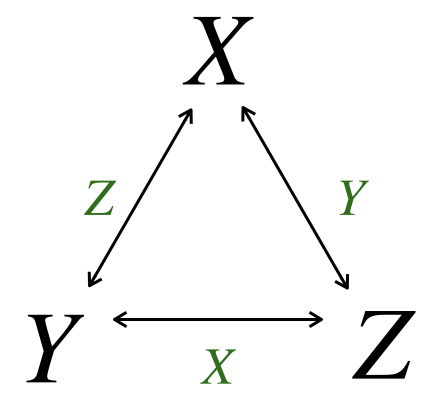}
	\caption{The application of $\lfloor \cdot,\cdot \rceil$ to operators in $\mathcal{B}$.}
	\label{PauliCommute}		
\end{figure}
\begin{figure}		
	\includegraphics[width=0.6\textwidth,center]{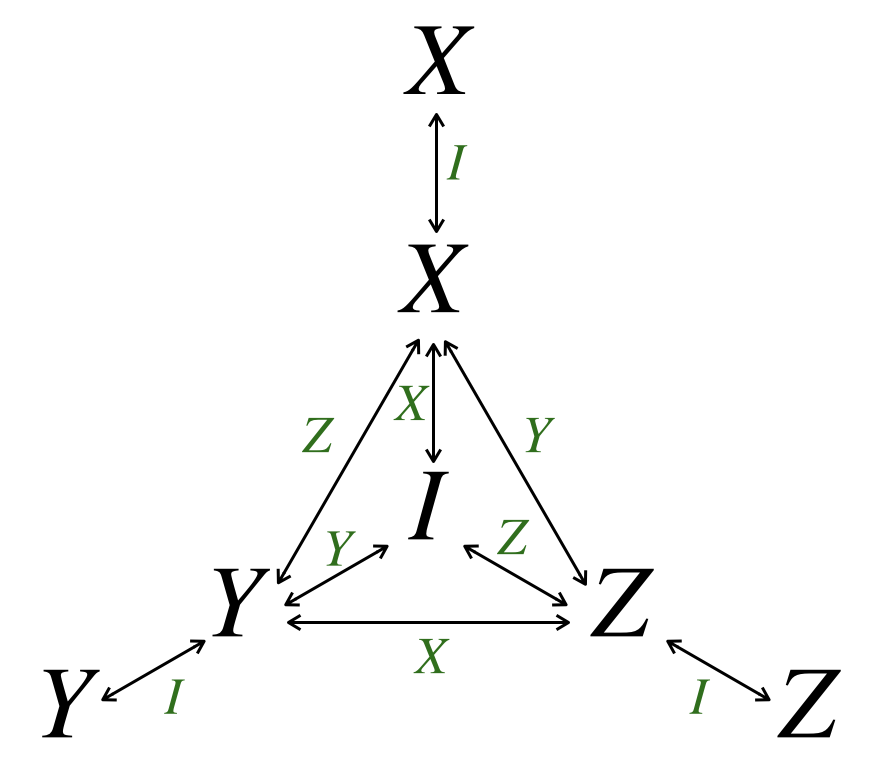}	
	\caption{The application of $\boxdot$ to a set of operators in $\mathcal{B}$.}
	\label{PauliMultip}		
\end{figure}
\begin{figure}		
	\includegraphics[width=0.35\textwidth,center]{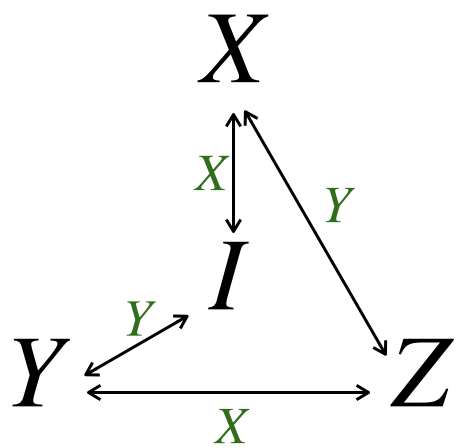}		
	\caption{The application of $\lfloor\cdot,\cdot\rceil$ and $\boxdot$ to operators in $\mathcal{B}$ while the edging operators are choosing from $\{X,Y\}$.}
	\label{PauliXY}		
\end{figure}

%==================Lemma: Change position===============================
\begin{lemma}\label{sublemma2}
	Given $O_1\in \Omega$ and $\nu_1,\nu_2,\cdots,\nu_m \in \bar{\digamma}$ as in \eqref{ExchangeDigamma} where $m$ is any natural number. Define two edging sequences 
	\begin{equation}\label{edges}
	\begin{split}
	E_1&=(\nu_{\epsilon(1)},\nu_{\epsilon(2)},\cdots,\nu_{\epsilon(m)})\\
	E_2&=(\nu_{\epsilon'(1)},\nu_{\epsilon'(2)},\cdots,\nu_{\epsilon'(m)})
	\end{split}
	\end{equation}
	where $\epsilon(\cdot)$ and $\epsilon'(\cdot)$ are two choices of the bijective map 
	$f \coloneqq \{1, 2, \cdots, m \} \rightarrow \{1, 2, \cdots, m\}$.
	For the two paths $\{O_1,E_1,O_2\}$ and $\{O_1,E_2,O_3\}$, we have $O_2=O_3$ if $O_2,O_3\in \Omega$.
\end{lemma}
\begin{proof}
	We first prove that Lemma \ref{sublemma2} holds for $O_1\in \mathcal{B}$ and $\bar{\digamma}=\{X,Y\}$. Then we generalize the conclusion to the case where $O_1\in \Omega$  and $\bar{\digamma}$ is as given in \eqref{ExchangeDigamma}.
	
	Note that the graph in Fig. \ref{PauliCommute} is a subgraph of the graph in Fig. \ref{PauliMultip}, we only consider the graph in Fig. \ref{PauliMultip}. Since the cell set for $\bar{\digamma}$ is $\{X,Y\}$, the graph in Fig. \ref{PauliMultip} can be simplified to the graph in Fig. \ref{PauliXY} by deleting all of the edges labeled by $I$ and $Z$. 
	
	We first prove that for $O_1\in \mathcal{B}$, $E_1$ and $E_2$ defined in \eqref{edges} and $\nu_1,\nu_2,\cdots,\nu_m \in \{X,Y\}$, we have $O_2=O_3$ where $O_2$ and $O_3$ are the two ending vertices of the two paths $\{O_1,E_1,O_2\}$ and $\{O_1,E_2,O_3\}$ such that $O_2,O_3\in \Omega$. This can be verified by the graph in Fig. \ref{PauliXY}. Note that graph in Fig. \ref{PauliXY} holds for both the operations $\lfloor \cdot,\cdot\rceil$ and $\boxdot$.
	Starting from any vertex in the graph in Fig. \ref{PauliXY}, for two paths $\{O_1,(X,Y),O_2\}$ and $\{O_1,(Y,X),O_3\}$, we have $O_2=O_3$ which indicates that the order of adjacent edging operators $X$ and $Y$ can be changed while the same ending vertices can be obtained. 
	
	Note that all of the possible bijective mapping $f$ can be realized if the changing of two arbitrary adjacent edges is allowed. Thus we exclude the case $O_2=0$ and $O_3=0$ and assume that $O_2,O_3\in \Omega$. Since the two edging sequences $E_1$ and $E_2$ are sequences of cell operators from $\{X,Y\}$ and we already show that the order of any two adjacent cell operators in the sequences is changable without changing the ending vertex, we can change the order of elements in $E_2$ to the same as $E_1$ while the ending vertex remains the same. Thus, for the two paths $\{O_1,E_1,O_2\}$ and $\{O_1,E_2,O_3\}$ where $E_1$ and $E_2$ are defined as in \eqref{edges}, we have $O_2=O_3$ if $O_2,O_3\in \Omega$.
	
   Now we consider the case where $O_1\in \Omega$  and $\bar{\digamma}$ is given in \eqref{ExchangeDigamma}. From \eqref{APPlemma04} and the definition of the operator $\lfloor\cdot,\cdot \rceil$ in Proposition \ref{CommuteOmega}, we see that the operation $\lfloor\cdot,\cdot \rceil$ can be decomposed as operations $\lfloor\cdot,\cdot \rceil$ and $\boxdot$ on cell operators. The above analysis shows that for operations $\lfloor\cdot,\cdot \rceil$ and $\boxdot$ on cell operators, the order of adjacent edges can be replaced without changing the ending vertex, given that the ending vertex is in $\Omega$. Thus, one can conclude that changing the adjacent edges in a sequence of edges $\nu_1,\nu_2,\cdots,\nu_m \in \bar{\digamma}$ has no influence on the ending vertex on the path, given that the ending vertex is in $\Omega$. By changing adjacent edges a finite number of times, one can realize any given mapping $f$. Thus, Lemma \ref{sublemma2} is proved. 
\end{proof}

%========================Lemma: delete even same edges===================
\begin{lemma}\label{sublemma3}
%For the operation $\lfloor \cdot,\cdot \rceil$, 
Let the edging sequence connecting $O_a\in \Omega$ and $O_b\in \Omega$ be denoted as $E$ and assume that $E\in S(\bar{\digamma}_{k+1})$ where $\bar{\digamma}_{k+1}$ is given in \eqref{ExchangeDigamma}. For an edging sequence $E'$ such that $C(E')=C(E)-C(E_p)$ where each element in $E_p\in S(\bar{\digamma}_{k+1})$ appears an even number of times, if there exists a path $\{O_a,E',O_c\}$ and $O_c\in\Omega$, then we have $O_a=O_c$.
\end{lemma}
Here, the subtraction $A-B$ for two collections $A$ and $B$ is defined as removing all of the elements in $B$ from $A$. 
\begin{proof}
	 Note that we have $E\in\bar{\digamma}$, which indicates that the cell operators for $E$ are $X$ and $Y$. We first prove that Lemma \ref{sublemma3} holds under the case where $O_a,O_b\in \mathcal{B}$ and $E,E_p\in S(\{X,Y\})$ while other statements remains unchanged.
	 
	 Since $E\in S(\{X,Y\})$, which indicates that there are no edging operators $Z$ and $I$, we delete edges that are labeled by $Z$ and $I$ from the graphs in Fig. \ref{PauliCommute} and Fig. \ref{PauliMultip}. Thus, the two graphs can then be simplified to the graph in Fig. \ref{PauliXY}. We take the cell operator $X$ as an example and the case $Y$ is equivalent. Given any edging sequence $E$ such that the path $\{O_a,E,O_b\}$ exists. Note that vertices in a path may not be distinct. For example, there exist two paths $\{X,(X,Y),Y\}$ and $\{X,(X,Y,X,Y,X,Y),Y\}$ connecting the two vertices $X$ and $Y$. Note that the basis elements $X$ and $Y$ can be applied to any vertex in the graph which means that an arbitrary ordering of edging operators selected from $\{X,Y\}$ can be applied to an arbitrary vertex in the graph in Fig. \ref{PauliXY}. According to Lemma \ref{sublemma2}, we then can change the ordering ofthe edging operators in the sequence $E$ at will and the path still $\{O_a,E,O_b\}$ exists. Let $\tilde{E}$ denote another edging sequence that shares the same collection of the edging operators with $E$ but with a different order of the edging operators such that all of the edging operators $X$ are placed before the edging operators $Y$ in the sequence. The path $\{O_a,\tilde{E},O_b\}$ exists according to Lemma \ref{sublemma2}. For example, $\tilde{E}=(X,X,X,Y,Y,Y)$ if $E=(X,Y,X,Y,X,Y)$. It can be shown that we have $O_s=O_e$ if there exist two paths $\{O_s,\{X,X\},O_e\}$ and $\{O_s,\{Y,Y\},O_e\}$. Thus, removing any pair of adjacent operators $X$/$Y$ from the path, the remaining edging sequence can still connect the starting and ending vertices. Thus, for any edging sequence $E'$ such that $C(E')=C(E)-C(E_p)$ where the elements in $E_p\in S(\{X,Y\})$ appear an even number of times, the path $\{O_a,E',O_b\}$ still exists.
	 
	 From \eqref{OperationDecomp}, we see that operation $\lfloor \cdot,\cdot\rceil$ on operators in $\Omega$ can be decomposed into operations $\lfloor \cdot,\cdot\rceil$ and $\boxdot$ on operators in $\mathcal{B}$. The conclusion obtained on the later case can be generalized to the former case. Thus, Lemma \ref{sublemma3} is proved.
	
	\iffalse
	We proof Lemma \ref{sublemma3} by contradiction. The assumption that $E$ connects $O_a$ and $O_b$ indicates that 
	There exists an ordering of elements in $E$ such that 
	\begin{equation*}
	\lfloor \sigma_a,\sigma_b \rceil, \sigma_c\rceil\neq 0.
	\end{equation*}
	Denote the number of elements in $E$ and $E'$ as $l$ and $m$, respectively. We have $m\leq l$. The fact that $E'$ can not connect $O_a$ and $O_b$ indicates that for an arbitrary ordering of edges in $E'$, there exists a number $p$ such that
	\begin{equation*}
	\lfloor \sigma_a,\sigma_b \rceil, \cdots\rceil, \sigma_p\rceil, E'_{p+1}\rceil=0,
	\end{equation*}
	which means $\lfloor \sigma_a,\sigma_b \rceil, \cdots\rceil, \sigma_p\rceil$ commuting with all edges from $(p+1)$th to $m$th. If this is the case, then there exists a number $q$ such that
	\begin{equation*}
	\lfloor \sigma_a,\sigma_b \rceil, \cdots\rceil, \sigma_q\rceil, E'_{q+1}\rceil=0,
	\end{equation*}
	which means $\lfloor \sigma_a,\sigma_b \rceil, \cdots\rceil, \sigma_p\rceil$ commuting with all edges from $(q+1)$th to $l$th. This contradict the assumption that $E$ can connect $O_a$ and $O_b$.
	
	Adding two same operators do not change the commutativity of a series of operators. For a series of Pauli matrix which always yields $\lfloor \sigma_a,\sigma_b \rceil, \cdots\rceil, \sigma_p\rceil=0$ for any ordering of $\varepsilon$, it takes the following form
	\begin{equation*}
		\{\sigma_a, \sigma_b, \sigma_c\}
	\end{equation*}
	\fi
\end{proof}

%================Lemma: two paths ===============================
\begin{lemma}\label{sublemma4}	
	Given $\bar{\digamma}_{k+1}$ as in \eqref{ExchangeDigamma}, $\bar{M}$ is a decomposed measurement set and the graph is $\mathbb{G}=\{G,\mathbb{E}\}$ where $G=f(\Omega,\bar{\digamma}_{k+1},\bar{M})$. If there exists a path $\{O_a,E,O_b\}$ where $O_a\in G_k$, $O_b\in G_{k+1}$ and $E=(\nu_1,\nu_2,\cdots)\in S(\bar{\digamma}_{k+1})$, and $O_b$ is not (k+1)-infinite, we have $O_b\in G_k$.	 
	%Moreover, $O_a\in G_k$ and all of the vertices in the path except $O_b$ are in $G_{\lfloor k}\cap G_{\lfloor {k+1}}$.
\end{lemma}

\begin{proof}
	 Since $E\in S(\bar{\digamma}_{k+1})$ and the assumption that $O_b$ is not $(k+1)$-infinite, $O_b$ must be of the form 
	\begin{equation}\label{applemma604}
	i_j^k\in
	\begin{cases}
	\{0,1,2,3\}, \quad 1<j\leq k,\\
	\{0\},\qquad \quad j>k.
	\end{cases}
	\end{equation}
	The cell operators on the $(k+1)$th position for both $O_a$ and $O_b$ are $I$. This can only be achieved through the pattern in Fig. \ref{PauliMultip}, which is then simplified to Fig. \ref{PauliXY} for our case, but impossible for the pattern of Fig. \ref{PauliCommute}. In Fig. \ref{PauliXY}, it can be testified that there must be an even number of the edging operators $X$ and $Y$ in a path starting from the vertex $I$ and ending at $I$. It can then be generalized to the case that the edging operators $X_{k}X_{k+1}$ and $Y_{k}Y_{k+1}$ appear even times in $E$. 
	 
	 We obtain the collection $C_E'$ by removing all of the edging operators $X_{k}X_{k+1}$ and $Y_{k}Y_{k+1}$ from the collection $C(E)$.
	 
%	{\color{blue}{Then we should prove that there exists a path $\{O_a,E',O_c\}$ where $O_c\in\Omega$.}} We prove this by contradiction.

	  According to Lemma \ref{sublemma3}, we have $O_b=O_c$. Moreover, if for all of the edging sequences $E'$ such that $C(E')=C_E'$ and the triplet $\{O_a,E',O_c\}$, we have $O_c=0$. Then we have $O_b=0$ which contradicts the assumption that $O_b$ is on the path $\{O_a,E,O_b\}$. Then, the existence of the path $\{O_a,E',O_b\}$ can be confirmed. Thus, $O_b\in G_k$ since $E'\in S(\bar{\digamma}_k)$.
\end{proof}

%% ============ Proof of proposition ==========================

We now move to the proof of Proposition \ref{SubsetsLemma} using the previous lemmas. 
\begin{proof}

The proof of Assertion 1 is straightforward. For $\forall 1\leq l\leq \mu$,
 given that $G_l=f(\Omega, \bar{\digamma}_l, \bar{M})$  and $G_\mu=f(\Omega, \bar{\digamma}_\mu, \bar{M})$, since $\bar{\digamma}_{\lfloor l}\subseteq \bar{\digamma}_{\lfloor \mu}$ we have $G_l\subseteq G_\mu$. 
 
 Then we prove Assertion 2 and Assertion 3 at the same time. Essentially, Assertion 2 and Assertion 3 together state that $G_{\lfloor k}$ is $k$-finite and, moreover, the graph generated by $G_{\lfloor k}$ is connected.

 We use the induction method to prove Assertion 2 and Assertion 3. Suppose that $G_{\lfloor i}$ is $i$-finite and is connected for $\forall 1\leq i\leq k$, we prove that $G_{\lfloor k+1}$ is $(k+1)$-finite and $G_{\lfloor k+1}$ is connected. According to Lemma \ref{connectedgraph01}, $G_{k+1}$ is connected. Since $G_{k+1}=G_k+G_{\lfloor k+1}$ and $G^2_{\lfloor k+1}\in G_{\lfloor k+1}$ is not adjacent to $G_k$, $G^2_{\lfloor k+1}$ must be connected to $G^1_{\lfloor k+1}$. Thus, to prove Assertion 3, it suffices to prove that both the induced subgraphs $\mathbb{G}^1_{\lfloor k+1}$ and $\mathbb{G}^2_{\lfloor k+1}$ are connected given that $\mathbb{G}_k$ is connected. To prove Assertion 2, we need to prove that operators in $G^1_{\lfloor k+1}$ and $G^2_{\lfloor k+1}$ are $(k+1)$-finite given that $G_k$ is $k$-finite.

%First, we prove Assertion 2 by contradiction, which states that the set $G^2_{\lfloor k+1}$ is $(k+1)$-finite. We assume that there exists an operator $O_m\in G_{\lfloor k+1}$ which is not $(k+1)$-finite. Then we prove that $O_m\in G_k$ which contradicts the assumption $O_m\in G_{\lfloor k+1}$. 

Lemma \ref{Ass2&3G101} concerns the problem where one element operator in $G_{\lfloor k}$ generates two element operators in $G^1_{\lfloor k+1}$ and Lemma \ref{Ass2&3G1} concerns the problem where one element operator in $G_{\lfloor k}$ generates only one element operator in $G^1_{\lfloor k+1}$. For both cases, the generated operators are connected and $(k+1)$-finite. Then, from Lemma \ref{Ass2&3G101}, Lemma \ref{Ass2&3G1} and the assumption that Assertions 2 and 3 hold for $G^1_{\lfloor k}$, Assertions 2 and 3 hold for $G^1_{\lfloor k+1}$. Since $G^2_{\lfloor k+1}$ can be generated by $G^1_{\lfloor k+1}$, the set $G_{\lfloor k+1}$ is connected, which means Assertion 3 holds for $G_{\lfloor k+1}$. 

Having proved that Assertion 2 holds for $G^1_{\lfloor k+1}$, we prove that Assertion 2 also holds for $G^2_{\lfloor k+1}$. For $O_b\in G_{\lfloor k+1}$, there must be a path $\{O_a,E,O_b\}$ where $O_a\in G_k$ and all of the vertices in the path except $O_b$ are in $G_{\lfloor k}\cap G_{\lfloor {k+1}}$ since we have the assumption that $G_k$ is connected and previously proved that $G^1_{k+1}$ is connected. Then from Lemma \ref{sublemma4}, if $O_b$ is not $(k+1)$-finite, we have $O_b\in G_{k}$, which contradicts the assumption that $O_b\in G_{\lfloor k+1}$. Then we conclude that the set $G^2_{\lfloor k+1}$ is $(k+1)$-finite. Thus Assertion 2 is proved.

So far, we have proved that $G_{\lfloor k+1}$ is connected and the set $G_{\lfloor k+1}$ is $(k+1)$-finite, given that Assertion 2 and Assertion 3 apply to $G_{k}$. In our case, we assume that $\mathbb{G}_1=\mathbb{M}$ is connected. Thus, we can always find a $k$ which validating Assertion 2 and Assertion 3. Thus, using the induction method, Proposition \ref{SubsetsLemma} is proved.
\end{proof}
\end{appendix}
%=======================Third class===========================================

%Finally, we prove that the third class of elements in $G_{\lfloor k}$ obey the Assertions 2 and 3. The third class of operators are generated by operators in $G_{\lfloor k}$ so they must be connected in the graph $\mathbb{G}_{\lfloor k}$. Moreover, since they are not connected with elements in $G$

%\end{proof}

\begin{equation}\label{APP:lessqubit}
G_k=\cup_{i=1}^k G_{\lfloor i}
\end{equation}
given that 
\begin{equation}\label{APP:lessqubit01}
G_N=\cup_{i=1}^N G_{\lfloor i}
\end{equation}
where $G_{\lfloor i}$ is defined in \eqref{seperate}. Equations \eqref{APP:lessqubit} and \eqref{APP:lessqubit01} together indicates that $G_{\lfloor k}\in G_{k} \quad \forall 1<k\leq N$.
We have $G_k\subset G_{k+1}$ since there is no deletion but only addition of operators in $G_{k+1}$ compared with $G_k$ as the number of qubits in a chain system  with Hamiltonian given in \eqref{hamiltonianmodel} increases from $k$ to $k+1$. Then we have to prove that $G_k\cup G_{\lfloor k+1} =G_{k+1}$ and $G_k\cap G_{\lfloor k+1} =\emptyset$ for $1<k\leq N-1$. According to the definition of $G_k$, we have $G_k\cap G_{\lfloor k+1} =\emptyset$ and then \eqref{APP:lessqubit} is equivalent to the following equation
\begin{equation}\label{APP:01}
G_{\lfloor k+1}=f(\Omega,\bar{\digamma}_{\lfloor k+1} ,G_k)-G_k.
\end{equation}
For the exchange model without transverse field, we have $\bar{\digamma}_{k+1}=\bar{\digamma}_k \cup \bar{\digamma}_{\lfloor k+1}$ where $\bar{\digamma}_k$ is the Hamiltonian set of a chain system with $k$ qubits, $\bar{\digamma}_{\lfloor k+1}=\{X_kX_{k+1},Y_kY_{k+1}\}$ and $\bar{\digamma}_{k+1}$ is the Hamiltonian set for a chain system with $k+1$ qubits. 

We then prove that there is always a path connecting  
%We prove \eqref{APP:01} by contradiction. The situation where equation \eqref{APP:01} does not hold is that $f(\Omega, \{X_kX_{k+1},Y_kY_{k+1}\},G_k)$ contains operators that belongs to $G_{\lfloor k}$ but not belongs to $G_k$.  Assume that there exists an operator $O^a$ that $O^a \in G_{\lfloor k} $ while $ O^a \notin G_k$. Since the graph $\mathbb{G}$ is connected according to Lemma \ref{connectedgraph01}, there must be at least an operator $O^b\in G_k$ and a path connecting $O^a$ and $O^b$. Since operators in $\bar{\digamma}_{\lfloor k+1}$ only act on the $k$th and $(k+1)$th qubits, $O^a$ and $O^b$ take the following form
%\begin{equation}
%O_a=O_{(1,k-1)}\sigma_m, \quad O_b=O_{(1,k-1)}\sigma_n,
%\end{equation}
%where $m,n\in \{1,2,3\}$. In Figure \ref{PropositionDivide01}, we list all of the possible cases of choosing different $O_a$ and $O_b$. It can be seen that there exists no such a pair of operators $O_a$ and $O_b$ that are connected. Thus there exists no such operator $O_a\in G_{\lfloor i}$ while $O_a\notin G_k$. We have $G_{\lfloor k}\subset G_k, \forall 1\leq k<N$ and \eqref{APP:01} is proved.

Second, we prove the second part of Proposition \ref{SubsetsLemma}.
\fi
	
%	@book{wilson1996introduction,
	%	title={Introduction to Graph Theory},
	%	author={Wilson, R.J.},
	%	isbn={9780582249936},
	%	lccn={gb96022704},
	%	url={https://books.google.com.au/books?id=tSolAQAAIAAJ},
	%	year={1996},
	%	publisher={Longman}
%	}

\end{document}